%% file: main.tex
\documentclass{article}
\usepackage{amsmath,amssymb,amsthm,amsfonts,latexsym,mathtools,mathdots,bbm,graphicx,float}
\usepackage{bbold}
\usepackage[letterpaper,margin=1in]{geometry}
\usepackage[backref,colorlinks,citecolor=blue,bookmarks=true]{hyperref}
\usepackage[nameinlink]{cleveref}
\usepackage{mdframed}
\usepackage[dvipsnames]{xcolor}
\usepackage{thmtools}
\usepackage{thm-restate}
\usepackage[most]{tcolorbox}
\usepackage{relsize}%

\usepackage{tikz}
\usetikzlibrary{fit}
\usepackage{caption}
\usepackage{subcaption}

\title{Amortized Relaxed Locally Decodable Codes}
\author{Jeremiah Blocki and Justin Zhang}
\date{}

\begin{document}

\include{macros}
\newtheoremstyle{break}
  {\topsep}{\topsep}%
  {\itshape}{}%
  {\bfseries}{}%
  {\newline}{}%
\theoremstyle{break}
\newtheorem{theorem}{Theorem}
\newtheorem*{remark}{Remark}
\newtheorem{lemma}[theorem]{Lemma}
\newtheorem{corollary}[theorem]{Corollary}
\newtheorem{definition}[theorem]{Definition}
\newtheorem{proposition}[theorem]{Proposition}
\newtheorem{construction}{Construction}

\maketitle

\begin{abstract}
 \input{sections/abstract}   
\end{abstract}

\section{Introduction}
\input{sections/introduction}

\section{Technical Overview}

\input{sections/technicalOverview}

\section{Preliminaries}\label{sect:Prelim}
\input{sections/preliminaries}

\section{Fuzzy Block Decoding Codes}\label{sect:FBDC}
\input{sections/fbdc}

\section{Amortized Correctability to Amortized Decodability}\label{sect:LCCtoLDC}
\input{sections/alccToAldc}

\section{All Together: Instantiating the Ideal Construction}\label{sect:AllTogether}
\input{sections/constconstruction}

\section{Amortized Locality Less Than  \texorpdfstring{$2$}{1}}\label{sect:constrApproach1}
\input{sections/approach1construction}

\bibliographystyle{IEEEtran}
\bibliography{abbrev0,crypto,other}

\appendix

\input{sections/appendix}

\section{Additional Discussion of Related Work}\label{sect:PriorWork}

\input{sections/priorwork}

\end{document}

%% file: macros.tex
\newcommand{\todo}[1]{\textcolor{red}{\underline{Todo}: #1}}
\newcommand{\jadd}[1]{\textcolor{PineGreen} {#1}}
\newcommand{\justin}[1]{\textcolor{blue} {\underline{Justin says:} #1}}
\newcommand{\jrem}[1]{\textcolor{purple} {\underline{Justin removed:} #1}}

\newcommand{\new}[1]{\textcolor{MidnightBlue}{#1}}
\newcommand{\jeremiah}[1]{\textcolor{red}{Jeremiah: #1}}

\newcommand{\msf}[1]{\ensuremath{\mathsf{#1}}}
\newcommand{\mc}{\mathcal}

\newcommand{\calA}{\mathcal{A}}
\newcommand{\calB}{\mathcal{B}}
\newcommand{\calC}{\mathcal{C}}
\newcommand{\calD}{\mathcal{D}}
\newcommand{\calE}{\mathcal{E}}
\newcommand{\calF}{\mathcal{F}}
\newcommand{\calG}{\mathcal{G}}
\newcommand{\calH}{\mathcal{H}}
\newcommand{\calI}{\mathcal{I}}
\newcommand{\calJ}{\mathcal{J}}
\newcommand{\calK}{\mathcal{K}}
\newcommand{\calL}{\mathcal{L}}
\newcommand{\calM}{\mathcal{M}}
\newcommand{\calN}{\mathcal{N}}
\newcommand{\calO}{\mathcal{O}}
\newcommand{\calP}{\mathcal{P}}
\newcommand{\calQ}{\mathcal{Q}}
\newcommand{\calR}{\mathcal{R}}
\newcommand{\calS}{\mathcal{S}}
\newcommand{\calT}{\mathcal{T}}
\newcommand{\calU}{\mathcal{U}}
\newcommand{\calV}{\mathcal{V}}
\newcommand{\calW}{\mathcal{W}}
\newcommand{\calX}{\mathcal{X}}
\newcommand{\calY}{\mathcal{Y}}
\newcommand{\calZ}{\mathcal{Z}}

\newcommand{\R}{\mathbbm R}
\newcommand{\C}{\mathbbm C}
\newcommand{\N}{\mathbbm N}
\newcommand{\Z}{\mathbbm Z}
\newcommand{\F}{\mathbbm F}
\newcommand{\Q}{\mathbbm Q}

\newcommand{\E}{\mathbbm E}

\newcommand{\bone}{\boldsymbol{1}}
\newcommand{\bbeta}{\boldsymbol{\beta}}
\newcommand{\bdelta}{\boldsymbol{\delta}}
\newcommand{\bepsilon}{\boldsymbol{\epsilon}}
\newcommand{\blambda}{\boldsymbol{\lambda}}
\newcommand{\bomega}{\boldsymbol{\omega}}
\newcommand{\bpi}{\boldsymbol{\pi}}
\newcommand{\bphi}{\boldsymbol{\phi}}
\newcommand{\bvphi}{\boldsymbol{\varphi}}
\newcommand{\bpsi}{\boldsymbol{\psi}}
\newcommand{\bsigma}{\boldsymbol{\sigma}}
\newcommand{\btheta}{\boldsymbol{\theta}}
\newcommand{\btau}{\boldsymbol{\tau}}
\newcommand{\ba}{\boldsymbol{a}}
\newcommand{\bb}{\boldsymbol{b}}
\newcommand{\bc}{\boldsymbol{c}}
\newcommand{\bd}{\boldsymbol{d}}
\newcommand{\be}{\boldsymbol{e}}
\newcommand{\boldf}{\boldsymbol{f}}
\newcommand{\bg}{\boldsymbol{g}}
\newcommand{\bh}{\boldsymbol{h}}
\newcommand{\bi}{\boldsymbol{i}}
\newcommand{\bj}{\boldsymbol{j}}
\newcommand{\bk}{\boldsymbol{k}}
\newcommand{\bell}{\boldsymbol{\ell}}
\newcommand{\bm}{\boldsymbol{m}}
\newcommand{\bn}{\boldsymbol{n}}
\newcommand{\bo}{\boldsymbol{o}}
\newcommand{\bp}{\boldsymbol{p}}
\newcommand{\bq}{\boldsymbol{q}}
\newcommand{\br}{\boldsymbol{r}}
\newcommand{\bs}{\boldsymbol{s}}
\newcommand{\bt}{\boldsymbol{t}}
\newcommand{\bu}{\boldsymbol{u}}
\newcommand{\bv}{\boldsymbol{v}}
\newcommand{\bw}{\boldsymbol{w}}
\newcommand{\bx}{{\boldsymbol{x}}}
\newcommand{\by}{\boldsymbol{y}}
\newcommand{\bz}{\boldsymbol{z}}
\newcommand{\bA}{\boldsymbol{A}}
\newcommand{\bB}{\boldsymbol{B}}
\newcommand{\bC}{\boldsymbol{C}}
\newcommand{\bD}{\boldsymbol{D}}
\newcommand{\bE}{\boldsymbol{E}}
\newcommand{\bF}{\boldsymbol{F}}
\newcommand{\bG}{\boldsymbol{G}}
\newcommand{\bH}{\boldsymbol{H}}
\newcommand{\bI}{\boldsymbol{I}}
\newcommand{\bJ}{\boldsymbol{J}}
\newcommand{\bK}{\boldsymbol{K}}
\newcommand{\bL}{\boldsymbol{L}}
\newcommand{\bM}{\boldsymbol{M}}
\newcommand{\bN}{\boldsymbol{N}}
\newcommand{\bP}{\boldsymbol{P}}
\newcommand{\bQ}{\boldsymbol{Q}}
\newcommand{\bR}{\boldsymbol{R}}
\newcommand{\bS}{\boldsymbol{S}}
\newcommand{\bT}{\boldsymbol{T}}
\newcommand{\bU}{\boldsymbol{U}}
\newcommand{\bV}{\boldsymbol{V}}
\newcommand{\bW}{\boldsymbol{W}}
\newcommand{\bX}{\boldsymbol{X}}
\newcommand{\bY}{\boldsymbol{Y}}
\newcommand{\bZ}{\boldsymbol{Z}}

\newcommand{\unif}{\overset{{\scriptscriptstyle\$}}{\leftarrow}}

\newcommand{\poly}{\text{poly}}
\newcommand{\polylog}{\text{polylog}}
\newcommand{\negl}{\text{negl}}

\newcommand{\Gen}{\ensuremath{\msf{Gen}}}
\newcommand{\Enc}{\ensuremath{\msf{Enc}}}
\newcommand{\Dec}{\ensuremath{\msf{Dec}}}
\newcommand{\sk}{\ensuremath{\msf{sk}}}
\newcommand{\pk}{\ensuremath{\msf{pk}}}

\makeatletter
\newcommand{\subalign}[1]{%
  \vcenter{%
    \Let@ \restore@math@cr \default@tag
    \baselineskip\fontdimen10 \scriptfont\tw@
    \advance\baselineskip\fontdimen12 \scriptfont\tw@
    \lineskip\thr@@\fontdimen8 \scriptfont\thr@@
    \lineskiplimit\lineskip
    \ialign{\hfil$\m@th\scriptstyle##$&$\m@th\scriptstyle{}##$\hfil\crcr
      #1\crcr
    }%
  }%
}
\makeatother

\renewcommand{\divideontimes}{\Cap}

  \newtcolorbox{notebox}[3][]{
  lower separated=false,
  top=0.2cm,
  grow to left by=-.25cm,
  grow to right by=-.25cm,
  colback=white,
  sharp corners,
  colframe={#3},
  fonttitle=\bfseries,
  colbacktitle=white,
  coltitle=black,
  enhanced,
  attach boxed title to top left={yshift=-0.25cm, xshift=0.6cm},
  boxed title style={boxrule=0pt, colframe=white, left=1pt, right=1pt},
  breakable,
  title={#2},
  {#1}
}

\newenvironment{betterFrame}[1]{
  \begin{notebox}{#1}{Aquamarine}
}{\end{notebox}}

\newenvironment{weirdFrame}[1]
  {\mdfsetup{
    frametitle={\colorbox{white}{\space#1\space}},
    innertopmargin=0pt,
    skipabove=\topskip, skipbelow=\topskip,
    frametitleaboveskip=-\ht\strutbox,
    frametitlealignment=\center
    }
  \begin{mdframed}%
  }
  {\end{mdframed}}

\newenvironment{weirdoFrame}[1]
  {\mdfsetup{
    frametitle={\colorbox{white}{\space#1\space}},
    innertopmargin=0pt,
    skipabove=\topskip, skipbelow=\topskip,
    nobreak=true,
    frametitleaboveskip=-\ht\strutbox,
    frametitlealignment=\center
    }
  \begin{mdframed}%
  }
  {\end{mdframed}}
\newcommand{\eps}{\varepsilon}
\newcommand{\inprod}[1]{\left\langle #1 \right\rangle}
\newcommand{\m}[1]{\begin{bmatrix}#1\end{bmatrix}}

\newcommand{\LDC}{\ensuremath{\msf{LDC}}}
\newcommand{\aLDC}{\ensuremath{\msf{aLDC}}}
\newcommand{\pLDC}{\ensuremath{\msf{pLDC}}}
\newcommand{\LCC}{\ensuremath{\msf{LCC}}}
\newcommand{\ham}{\msf{Ham}}
\newcommand{\Ham}{\ham}

\newcommand{\ed}{\msf{ED}}
\newcommand{\ED}{\ed}

\newcommand{\aldc}{\msf{aLDC}}
\newcommand{\paldc}{\msf{paLDC}}
\newcommand{\raldc}{\msf{raLDC}}
\newcommand{\paLDC}{\paldc}
\newcommand{\paldcgame}{\texttt{paLDC-Sec-Game}}
\newcommand{\raldcgame}{\texttt{raLDC-Game}}

\newcommand{\out}[1]{#1_{\text{out}}}
\newcommand{\inn}[1]{#1_{\text{in}}}
\newcommand{\insdel}[1]{#1_{\text{insdel}}}
\newcommand{\h}[1]{#1_{\text{h}}}
\newcommand{\p}[1]{#1_{\msf{p}}}
\newcommand{\res}[1]{#1_{\msf{r}}}
\newcommand{\calCout}{\out \calC}
\newcommand{\calCin}{\inn \calC}
\newcommand{\alphaout}{\out \alpha}
\newcommand{\alphain}{\inn \alpha}
\newcommand{\kappaout}{\out \kappa}
\newcommand{\kappain}{\inn \kappa}
\newcommand{\deltain}{\inn \delta}
\newcommand{\deltaout}{\out \delta}
\newcommand{\epsout}{\out \eps}
\newcommand{\epsin}{\inn \eps}
\newcommand{\betain}{\inn \beta}

\newcommand{\sz}[1]{#1_{\text{sz}}}
\newcommand{\calCsz}{\sz \calC}
\newcommand{\Encsz}{\sz \Enc}
\newcommand{\Decsz}{\sz \Dec}

\newcommand{\buf}{\text{buf}}
\newcommand{\NBS}{\texttt{Noisy-Binary-Search}}
\newcommand{\IntDec}{\texttt{Interval-Decode}}
\newcommand{\BlockDec}{\msf{Block}\text{-}\msf{Decode}}
\newcommand{\bufs}{\buf s}

\newcommand{\BufFind}{\texttt{Buff-Find}}

\newcommand{\Compile}{\texttt{Compile}}
\newcommand{\EncCompile}{\msf{EncCompile}}
\newcommand{\Extract}{\msf{Extract}}
\newcommand{\DecCompile}{\texttt{DecCompile}}

\newcommand{\RecoverBlock}{\texttt{RecoverBlock}}
\newcommand{\FindBlock}{\texttt{RecoverBlock}}
\newcommand{\FindBlocks}{\texttt{RecoverBlocks}}
\newcommand{\RecoverBlocks}{\texttt{RecoverBlocks}}
\newcommand{\Search}{\texttt{Search}}
\newcommand{\MarkGood}{\texttt{MarkGood}}
\newcommand{\BuffFind}{\texttt{BuffFind}}

\newcommand{\Good}{\msf{Good}}

\newcommand{\uDec}{\ensuremath{\msf{uDec}}}
\newcommand{\LDCwUGame}{\texttt{LDCwU-Game}}
\newcommand{\LDCwU}{\msf{LDCwU}}

\newcommand{\pLDCwUGame}{\texttt{pLDC-w-Uncert-Game}}
\newcommand{\pLDCwU}{\msf{pLDCwU}}

\newcommand{\findGood}{\texttt{findGood}}

\newcommand{\BI}{\calB\calI}

\newcommand{\size}[1]{\left| #1 \right|}
\newcommand{\simu}[1]{#1_{\msf{sim}}}

\newcommand{\simuu}[1]{#1_{\msf{sim,u}}}

\newcommand{\uDecGame}{\texttt{UDec-Game}}

\newcommand{\Bad}{\texttt{Bad}}

\newcommand{\reduce}{\texttt{reduce}}

\newcommand{\psk}{\msf{p},\sk}

\newcommand{\puzzgen}{\msf{PuzzGen}}
\newcommand{\PuzzGen}{\msf{PuzzGen}}
\newcommand{\puzzsolve}{\msf{PuzzSolve}}
\newcommand{\PuzzSolve}{\msf{PuzzSolve}}
\newcommand{\Sol}{\msf{PuzzSolve}}
\newcommand{\Puzz}{\msf{Puzz}}

\newcommand{\LdcStar}{\ensuremath{\LDC^*}}

\newcommand{\FBDC}{\ensuremath{\msf{FBDC}}}

\newcommand{\reaLDC}{\ensuremath{\msf{arLDC}}}
\newcommand{\arLDC}{\reaLDC}
\newcommand{\reaLCC}{\ensuremath{\msf{arLCC}}}
\newcommand{\realdc}{\reaLDC}
\newcommand{\realcc}{\reaLCC}
\newcommand{\reaLdc}{\reaLDC}
\newcommand{\reaLcc}{\reaLCC}
\newcommand{\rLdc}{\msf{rLDC}}
\newcommand{\rLcc}{\msf{rLCC}}
\newcommand{\rldc}{\rLdc}
\newcommand{\rLDC}{\rldc}
\newcommand{\rlcc}{\rLcc} 
\newcommand{\rLcC}{\rlcc}
\newcommand{\LTC}{\ensuremath{\msf{LTC}}}
\newcommand{\Ltc}{\LTC}
\newcommand{\ltc}{\LTC}

\newcommand{\Test}{\ensuremath{\msf{Test}}}
\newcommand{\Accept}{\ensuremath{\msf{T}}}
\newcommand{\rdist}{\ensuremath{\msf{reldist}}}
\newcommand{\reldist}{\rdist}

\newcommand{\vltc}{\ensuremath{\msf{VLTC}}}
\newcommand{\VLTC}{\vltc}

\newcommand{\Rate}[1]{\msf{R}_{#1}}
\newcommand{\rate}[1]{\Rate{#1}}

%% file: sections/abstract.tex
Locally decodable codes (LDCs) enable recovery of any message symbol by probing only a small number of positions in a possibly corrupted codeword. The central parameters of an LDC are its rate, locality, and error tolerance. Ideally, one would like all three parameters to be constant. However, classical lower bounds show that such codes cannot exist. A recent line of work introduced amortized locally decodable codes (aLDCs), in which the decoder is tasked with recovering an entire block of consecutive message symbols rather than a single symbol. The amortized locality is defined as the total number of codeword queries divided by the number of recovered symbols. While prior work obtained ideal aLDCs with constant rate, constant error tolerance, and constant amortized locality, those constructions relied on either shared randomness hidden from the channel or computational assumptions restricting the channel. 

Another well-studied relaxation is the notion of a relaxed locally decodable code (RLDC), in which the decoder may output a special failure symbol $\bot$ rather than risk decoding incorrectly. Even with recent breakthroughs, the best known RLDCs achieve only polylogarithmic locality, though this is dramatically better than what is known for standard LDCs without relaxed decoding; achieving constant locality with constant rate is provably impossible. In this work, we introduce the notion of an amortized relaxed locally decodable code (aRLDC), combining amortized decoding with the relaxed decoding paradigm. Unlike prior ideal aLDC constructions, our model is fully information-theoretic and makes no assumptions about shared randomness or computational limitations of the adversarial channel. We construct the first aRLDC with constant rate, constant error tolerance, and constant amortized locality. Moreover, for any block of length $\Omega(\mathrm{polylog}(k))$, our decoder achieves amortized locality $1+\delta^{1 - o(1)}$, where $\delta$ is the error tolerance parameter. Thus, asymptotically, recovering a long block requires essentially less than two codeword probe per message symbol recovered. By contrast, without amortization no RLDC can simultaneously achieve constant rate, constant error tolerance, and constant locality. 

%% file: sections/introduction.tex
A locally decodable code (\LDC{}) is an error-correcting code equipped with a decoding procedure that can recover any desired message symbol while reading only a small number of positions of a possibly corrupted codeword. Formally, an \LDC{} consists of an encoding procedure $\Enc:\{0,1\}^k\rightarrow \{0,1\}^n$ and a randomized local decoder $\Dec^{\bw}(i)$ that, given an index $i\in[k]$ and oracle access to a corrupted codeword $\bw$, outputs the message symbol $x_i$ with a sufficiently large constant success probability (e.g., at least $2/3$) while making at most $\ell$ queries. The central parameters of an \LDC{} are its rate $k/n$, its error tolerance $\delta$, and its locality $\ell$.

Ideally, one would like all three parameters to be constant with locality $\ell$ as small as possible. Unfortunately, this goal is provably unattainable. Katz and Trevisan showed that any \LDC{} with constant error tolerance $\delta=\Omega(1)$ and locality $\ell \geq 2$ must have codeword length $n \in \Omega(k^{\ell/(\ell-1)})$ \cite{STOC:KatTre00} --- they additionally showed that it is impossible to achieve locality $\ell < 2$  regardless of rate. Moreover, the best-known constructions are even farther from our ideal goal. For example, matching vector codes achieve constant locality, but require super-polynomial codeword length \cite{LDCSURVEY,STOC:Efremenko09}. The Hadamard code gives a $2$-query \LDC{}, but it has exponential codeword length. In fact, {\em any} $2$-query \LDC{} must have exponential codeword length \cite{STOC:KerdWo03}. Despite decades of work, the best known \LDC{}s with constant rate and constant error tolerance achieve locality only $\ell = \exp(O(\sqrt{\log n \log\log n}))$, which is subpolynomial but still super-polylogarithmic \cite{STOC:KMRS16}.

Several relaxations of the \LDC{} model have been proposed to circumvent these impossibility results and technical barriers, yielding substantial progress but still falling short of our ideal goal. For example, if the encoding and decoding procedures are allowed to share a secret key or the error channel is assumed to be resource-bounded (e.g., computationally bounded, space-time bounded, etc.), then there are constructions achieving constant rate, constant error tolerance, and polylogarithmic locality \cite{ICALP:OstPanSah07,C:HemOst08,HemOstStraWoo11,ISIT:BloBlo21,SCN:AmeBloBlo22,ITC:BloKulZho20}. These constructions substantially improve on the best known tradeoffs for classical \LDC{}s \cite{STOC:KMRS16}, but still fall short of the ideal goal of constant locality. 

Another well-studied relaxation is the notion of a {\em relaxed Locally Decodable Code} (\rldc{}) \cite{STOC:BGHSV04}, in which the decoder is permitted to output a special failure symbol $\bot$ if it detects corruption. Unlike relaxations based on shared randomness or resource-bounded channels, \rldc{}s place no restrictions on the adversarial channel and do not require shared randomness. Thus, \rldc{}s remain fully information-theoretic making them an attractive relaxation to study. 

Similarly, the best-known \rldc{} constructions also achieve constant rate, constant error-tolerance and polylogarithmic locality  \cite{ITCS:GurRamRot18,STOC:BGHSV04,SODA:ChiGurShi20,STOC:KumMon24,CCC:CohYan24}. A major breakthrough by Kumar and Mon introduced a recursive nesting framework based on locally testable codes, yielding the first \rldc{} with constant rate, constant error tolerance, and polylogarithmic locality $O(\log^{69} n)$ \cite{STOC:KumMon24}. Subsequently, Cohen and Yankovitz refined this framework using expander codes to obtain locality $\log^{2+o(1)} n$ while maintaining constant rate and error tolerance \cite{CCC:CohYan24}. Despite this dramatic progress, these constructions still fall short of the ideal goal of constant locality. Indeed, Gur and Lachish \cite{SIAM:GurLac21} and Dall'Agnon, Gur, and Lachish \cite{SIAM:DGL23} showed that ideal \rldc{}s do not exist: any \rldc{} with locality $\ell$ must have codeword length at least $n \geq k^{1+1/O(\ell^2\log^2\ell)}$. Furthermore, if one requires locality $\ell = 2$ and perfect completeness, then any \rldc{} must have exponential codeword length \cite{CCC:BBCGLZZ26} \footnote{Perfect completeness means that for any valid codeword $\by$, the (possibly adaptive) decoder $\Dec^{\by}(i)$ always outputs the correct message symbol $x_i$. The lower bounds of \cite{SIAM:GurLac21,SIAM:DGL23} do not require perfect completeness.}, and in fact, locality $\ell < 2$ \rldc{}s (one query \rldc{}s) with perfect completeness do not exist (see Appendix \ref{sect:onequeryRLDC}). 

To summarize, despite substantial progress, no known relaxation of the \LDC{} model achieves simultaneous constant rate, constant error tolerance, and constant locality. This raises the question of whether a fundamentally different relaxation can overcome these barriers.

Due to these impossibility results and barriers, another line of work has focused on amortizing locality across the recovery of multiple message symbols \cite{ISIT:RamWoo18,ISIT:blozha25,ITC:bloZha25,cramer2019efficient}. In the most general formulation, the decoder $\Dec^{\bw}: 2^{[k]} \rightarrow \{0,1\}^{k}$ is given a subset $S \subseteq  [k]$ of  $|S| \geq \kappa$ message coordinates as input and must recover all symbols in $\bx[S]$. The  locality is amortized over the recovered symbols: if the decoder makes $\ell$ codeword queries to recover $|S|$ message symbols, then its amortized locality is $\ell/|S|$. Cramer et al.~\cite{cramer2019efficient} demonstrated that local decoding of Reed--Muller codes can be amortized across multiple requested coordinates. However, their construction is tailored to Reed--Muller codes and does not achieve constant amortized locality in the constant rate/error tolerance regime. Thus, despite demonstrating the potential benefits of amortization, their work remains far from the ideal regime of constant rate, constant error tolerance, and constant amortized locality.

In many applications, it is natural to assume that the requested subset $S=[L,R]:=\{L,L+1,\ldots,R\}$ forms a contiguous interval, and we will work in this setting throughout the paper. Blocki and Zhang showed that \aldc{}s with constant rate, constant error tolerance, and constant amortized locality exist when the adversarial channel is resource-bounded and the requested coordinates form a contiguous interval $S = [L,R]$ \cite{ISIT:blozha25}. Thus, ideal amortized local decoding is possible under computational assumptions on the channel. However, their results do not apply in the information-theoretic setting where the adversarial channel is computationally unbounded. To the best of our knowledge, no ideal amortized local decoding construction is known in this setting. Indeed, it is not even known whether there exists an amortized \rldc{} with constant rate, constant error tolerance, and constant amortized locality.

We ask the following questions
\begin{quote}
    \begin{enumerate}
        \item Does there exist an information-theoretic amortized relaxed locally decodable code (\realdc{}) with constant rate, constant error tolerance, and constant amortized locality?
        \item Can one achieve amortized locality less than $2$ while keeping the rate and error tolerance constant, in contrast with the impossibility of $\ell < 2$ \rldc{}s? 
        Further, can we have the amortized locality approach $1?$
    \end{enumerate}
\end{quote}
In this work, we answer all questions in the affirmative. First, we construct the first information-theoretic \realdc{} with constant rate, constant error tolerance, and constant amortized locality. Second, we show how to refine this construction to achieve amortized locality $\frac{1}{1 - \delta^{1 - o(1)}} = 1 + \delta^{1 - o(1)} < 2$ while keeping constant rate and constant error tolerance.
Lastly, we show that by allowing the error tolerance to approach $0$, we can achieve amortized locality approaches $1$ with constant rate.

To state our results more formally, we call a code $\calC=(\Enc,\Dec)$ an $(\alpha,\kappa,\delta,\eps)$-\realdc{} if, for every interval $S=[L,R]\subseteq [k]$ of size $|S|\ge \kappa$, every message $\bx\in\Sigma^k$, and every corrupted codeword $\bw$ satisfying $\reldist(\bw,\by):= \frac{\left|\left\{ i: \bw[i] \neq \by[i] \right\} \right|}{n} \le \delta$ from the corresponding codeword $\by=\Enc(\bx)$, the decoder $\Dec^{\bw}(S)$ makes at most $\alpha|S|$ queries to $\bw$ and outputs $\bx[S]$ or $\bot$ with probability at least $1-\eps$. Our main results are summarized by the following informal theorems.

\begin{theorem}[Ideal $\realdc{}$, informal]\label{ithm:constantARLDC}
    For any constant $\eps < 1$, there exists constant $c = c(\eps)$ such that there exists a {$(\alpha = O(1),\kappa = c\times (\log n)^2 \times 2^{(\log\log\log n)^4},\delta = \Omega(1),\eps)$-\realdc{}} with constant rate.
\end{theorem}

\begin{theorem}[Ideal $\realdc{}$ with Amortized Locality Below $2$ (or Approaching $1$), informal]\label{ithm:approachingOneARLDC} 
    For any constants $\nu >0, \eps < 1$, there exist $\delta = \delta(\nu)$ and constant $c = c(\nu,\eps)$ such that there exists a ${(\alpha = 1+\nu, \kappa = \frac{c}{\nu} \times (\log n)^2 \times 2^{(\log\log\log n)^4 }, \delta, \eps)}$ \realdc{} with rate at least $1 - \delta\times2^{O((\log\log\log1/\delta)^3)}.$ 
\end{theorem}
Note that  $(\log n)^2 \times 2^{(\log\log\log n)^4} \leq \log^{2 + o(1)} n$ and $(\log n)^2 \times 2^{(\log\log\log n)^4} \geq (\log n)^2 \times (\log\log n)^{O(1)}.$
We remark that for any $\nu > 0$ we can achieve amortized locality $1 + \nu$ while tolerating $\delta = \nu\times 2^{O(\log\log\log \nu)^3)} = \nu^{1 - o(1)}$  fraction of errors.
Hence, for fixed $\nu < 1$, we can achieve amortized locality less than $2$ with constant rate and constant error tolerance. 
Further, by allowing $\nu$ to approach $0,$ we can have the amortized locality and rate approach $1$.

%% file: sections/technicalOverview.tex
In this section, we give a technical overview of our \realdc{} constructions. 
We recall the nested \rlcc{} constructions that serve as a starting point to our \realdc{} constructions, detail the challenges in amortizing the locality of prior constructions, outline our contributions which overcome the discussed challenges, and give a detailed summary on how we precisely construct our final \realdc{} constructions.

\subsection{Recalling the Nested \LCC{} Construction}\label{subsec:recallingLCC}
A natural starting point for our amortized construction is the nested code \rlcc{} constructions, devised by Kumar and Mon and improved upon by Cohen and Yankovitz \cite{STOC:KumMon24,CCC:CohYan24}. 
For ease of exposition, we discuss the construction with parameters guaranteed by the probabilistic method.
The explicit construction follows the same ideas with more complicated parameters.
Recall that \rlcc{} decoders are tasked with recovering codeword symbols, while \rldc{} decoders are tasked with recovering message symbols. However, linear \rlcc{}s may be transformed into linear \rldc{}s with the same rate, distance, and locality so the constructions of \cite{STOC:KumMon24,CCC:CohYan24} yield \rldc{}s with constant rate, distance and poly-logarithmic locality. The nested \rlcc{} constructions rely on the notion of a $\delta$-locally testable code\footnote{We follow the more restricted form of local testing used by Cohen and Yankovitz, which supposes an upperbound on the relative distance.}. Informally, a $(\ell_T,\delta,\rho)$-locally testable code admits a local testing algorithm $\Test$ which makes at most $\ell_T$ queries to a noisy codeword $\bw$. The guarantee is that if $\reldist(\bw,\by) \leq \delta$ then $\Test{}^{\bw}$ returns $\bot$ with probability at least $\rho \times \reldist(\bw,\by)$ --- if $\reldist(\bw,\by) = 0$ then $\Test{}^{\bw}$ should always accept.

Given a codeword $\by$ we can partition $\by$ into $2^j$ blocks $\by = \by^{j}_1 \circ \ldots \circ \by^{j}_{n/2^j}$ for any $j \leq s$ where $s=O(\log n)$ denotes the total number of nested levels in the construction. This forms a natural binary tree structure {of height $s$}, where each node of the binary tree on level $j \in [s]$ represents a contiguous block {of size $(n/2^j)$} that is locally testable. We illustrate this structure in Figure~\ref{fig:blockstruct}, denoting each intermediate codeword block as $\by_{j}^i$. Intuitively, the nested code construction ensures that every block $\bw^{j}_i$ is a $(\ell_T',\delta' = o(1),\rho')$-locally testable code for all blocks $i \leq n/2^j$ at all levels $j \leq s$ (As a special case, when $j=0$ we have $\bw^{0}_1 = \by$ and we require that the code is $(\ell_T,\delta,\rho)$-locally testable for a larger $\delta = \Omega(1)$ and $\rho = O(1/\delta)$). Crucially, $s$ is set such that $n/2^s < 1/\delta'$, i.e., the lowest level block have size less than $1/\delta'$ for correctness (discussed below).


\input{tikz/binaryStructure}

Suppose we are given a corrupted codeword $\bw$ with relative distance $\reldist(\bw,\by) \leq \delta$ from the original codeword $\by$. Observe that $\bw = \bw^0_1 = \bw^{1}_1 \circ \bw^1_2 = \ldots = \bw^{s}_1 \circ \ldots \circ \bw^s_{n/2^s}$. The $\rlcc{}$ decoder $\Dec^{\bw}(i)$ will begin by repeatedly running the local tester $\Test{}^{\bw}$ on the entire word $\bw = \bw^0_1$. If $\Test{}^{\bw}$ ever outputs $\bot$ then the decoder may also output $\bot$; otherwise the decoder can proceed assuming, with high confidence, that $\reldist(\bw,\by) \leq \delta'/2$ (This implies that $\reldist(\bw^1_1,\by^1_1) \leq \delta'$ and $\reldist(\bw^1_2,\by^1_2) \leq \delta'$ for both blocks $\bw^1_{1}$ and $\bw^1_{2}$ on the next level). Then, the $\rlcc{}$ decoder $\Dec^{\bw}(i)$, on input index $i \in [n]$, computes a path $(i_0=1,i_{1},i_2,\dots,i_{s})$ from the top level $0$ to the lowest level block $\bw^{s}_{i_s}$, where for each $j \leq s$ the block $\bw^{s}_{i_s}$ contains the symbol  $\bw[i]$. In particular, $i_s$ is the index of the lowest level testable block $\bw^{s}_{i_s}$ containing symbol $\bw[i]$. For each level $1 \leq j \leq s$ the decoder $\Dec^{\bw}(i)$ runs the local tester $\Test{}^{\bw^j_{i_j}}$ on the block $\bw^j_{i_j}$. If $\Test{}^{\bw^j_{i_j}} = \bot$ for any $1 \leq j \leq s$ (any local test fails) then the decoder $\Dec^{\bw}(i)$ will output $\bot$. Otherwise, if all local tests pass, the decoder will output the bit $\bw[i]$. 

Intuitively, {correctness follows from inductively guaranteeing (with high probability) that by the point we are testing block $i_j$, the fraction of errors is less than $\delta$, thereby satisfying the precondition for correctness of the local tester.}
After the last test passes on the lowest level block corresponding to index $i_s$, we guarantee that with high probability that there is at most $\delta' \times (n /2^s) < 1$ errors, i.e., there are no errors in the block $i_s$ containing $\bw[i].$
Thus, we can safely read and output $\bw[i] = \by[i].$ 
Conversely, if $\bw[i] \neq \by[i]$, then $\reldist(\by^s_{i_s},\bw^s_{i_s}) \geq \delta'$ since there $\bw^s_{i_s}$ contains at most $(n / 2^s) < 1/\delta'$ symbols by our choice of the nesting parameter $s$. In this case there must be some level $j$ where $\delta'/2 \leq \reldist(\by^j_{i_j},\bw^j_{i_j}) \leq \delta'$. In this case the local tester $\Test{}^{\bw^j_{i_j}}$ will output $\bot$ with high probability. 


\subsection{Technical Challenges} \label{subsec:challenges}
{We discuss natural ideas to amortizing this construction, and explain their shortcomings. We first observe that the described construction immediately allows us to recover {\em all} the codeword bits in the lowest level block $\bw_{i_s}^s$ after running the local testing algorithm at each level. In particular, if the local testing algorithm never outputs $\bot$ then we can conclude that there are no errors in the entire block  $\bw_{i_s}^s$.}

\paragraph{Challenge $1$: Small block size}
{
Since the total number of queries made is $\ell_T + s\times \ell_T' + 1 =  O(\log^2 n\times \log\log n)$ to do the local testing on blocks $i_0, i_1,\dots,i_s$ and to output bit $w_i$, the amortized locality $\alpha$ would be 
$\alpha = \frac{\overbrace{O(\log^2 n\times \log\log n)}^{\text{local testing queries}} + |\bw_{i_s}^i|}{|\bw_{i_s}^i|}.$
Unfortunately, we cannot set the lowest level block size $|\bw_{i_s}^i|$ to ensure that $\alpha=O(1)$ is constant. To ensure that the lowest level block $\bw_{i_s}^s$ does not have {\em any} errors we required that $|\bw_{i_s}^s | \leq 1/\delta'$. When the original nested code construction is instantiated with the best known expanders (via the probabilistic method) we have $\delta' = \Omega\left(\frac{1}{\log n \log \log n}\right)$. Then, the amortized locality $\alpha \geq \Omega(\log n)$ since $|\bw_{i_s}^s | \leq 1/\delta' \leq O\left(\log n \log \log n\right)$. Thus, a naive instantiation of the nested code construction does not achieve  constant amortized locality.
}

{Since the amortized locality is bottle-necked by the lowest level block size, the next natural idea is to increase the lowest level block sizes.
For constant amortized locality, it would suffice to set $s$ to be a smaller value such that the lowest level block size is within the order of the number of queries $O(\log^2 n \times \log \log n),$ i.e., set $s$ such that $n / 2^s \geq \Omega(\log^2 n \times \log\log n).$
Unfortunately, by gaining constant amortized locality, we lose soundness since block $\bw_{i_s}^s$ may contain errors since we increased the block size, $|\bw_{i_s}^s| \gg 1/\delta$. 
{Previously, setting the lowest level block size $|\bw_{i_s}^s| \leq 1/\delta$, allowed us to conclude that either the lowest level block was error free, or any error would have been caught by a local tester at some level $j \in [s]$. Thus, by setting $|\bw_{i_s}^s| \gg 1/\delta$, soundness is no longer guaranteed, and we only achieve super-constant amortized locality without modifying the code.}}

\paragraph{Challenge $2$: Amortized Locality Preserving \realcc{} to \realdc{} Transformation} Another challenge is the conversion of the amortized locally {\em correctable} code into an amortized locally {\em decodable} code.
It is well known that one can convert a linear \rlcc{} into a linear \rldc{} via a linear transformation on the code that embeds the message symbols into the codewords i.e., there exists indices $i_1,\ldots, i_k$ such that $\bx[j] = \by[i_j]$ for each message symbol $j \in [k]$. This transformation additionally preserves the rate, error tolerance, and locality 
(notably, the linear transformation does not change the codeword space, so the \rldc{} decoder can simply run the \rlcc{} decoder $\Dec^{\bw}(i_j)$ to recover $\bx[j]$).
However, this linear transformation may not necessarily preserve {\em amortized} locality. 

There are two main issues that may impact amortized locality. First, the transformation may not ensure that message symbols appear in consecutive order i.e., we have $i_j > i_{j+1}$. Second, the message symbols may not be uniformly distributed throughout the codeword i.e., we may have some blocks $\by^s_j$ which contain relatively few message symbols i.e., $\left| \by^s_j \cap \{i_1,\ldots, i_k\} \right| \ll n/2^s$. The first issue is not too difficult to address since we can always permute the symbols of our message $\bx$ before encoding to ensure that they appear in consecutive order after the linear transformation is applied. Addressing the second issue is a bit more challenging. In particular, when the amortized local decoder $\Dec$ is tasked with outputting all message symbols in the range $[L,R] \subseteq [k],$ we no longer have a guarantee on how many blocks of the nested codeword that the local decoder $\Dec$ will need to recover in order to extract the desired message symbols.

To summarize, without any modifications, ideal \realdc{}s are not immediate from prior work since (1) amortized locality is non-constant for correctness in the prior \rlcc{} construction, and (2) the transformation from linear \rlcc{}'s to linear \rldc{}'s does not necessarily preserve amortized locality since we cannot guarantee that each consecutive block contains a proportional number of message symbols. 

\subsection{Our Contributions}\label{subsec:contributions}

\paragraph{Achieving Constant Amortized Locality for Local Correctability} 
In this work, we first show how to increase the small block size $|\bw_{i_s}^s|$ of Cohen and Yankovitz's construction by replacing the expander code used for the lowest level blocks with a constant rate, constant error tolerant error correcting code $\calC_{\msf{block}} = (\Enc_{\msf{block}},\Dec_{\msf{block}})$.
Observe that if none of the testers fail on every other level, then we expect the error rate in each higher level block is at most $\delta'$.
This implies that the error rate for block $\bw_{i_s}^s$ is at most $\reldist(\by_{i_s}^{s},\bw_{i_s}^{s}) \leq 2\delta'$. We will ensure that that the lowest level block $\by_{i_s}^{s}$ is an error correcting code that can tolerate $2\delta'$ fraction of errors. Thus, if we are given $\bw_{i_s}^{s}$ with $\reldist(\by_{i_s}^{s},\bw_{i_s}^{s}) \leq 2\delta'$ we can recover the correct value  $\by_{i_s}^{s}$ by decoding $\bw_{i_s}^{s}$ and then re-encoding it with code $\calC_{\msf{block}}$ i.e., $\by_{i_s}^{s} = \Enc_{\msf{block}}(\Dec_{\msf{block}}(\bw_{i_s}^s))$.
Otherwise, if there are more than a $2\delta'$ fraction of errors, we can apply the same soundness argument in the prior construction to conclude that we would have outputted $\bot$ in some higher level (with high probability).

With this modification we are no longer constrained to pick $s$ such that $n/2^s < 1/\delta'$ since we can now tolerate errors in the lowest level blocks. By picking $s$ such that the lowest level block sizes to be at least $n/2^s = \Omega(\log^2n \times \log \log n)$ we can ensure that the amortized locality in our \realcc{} is constant. The local decoder works in essentially the same way until we reach the bottom layer where we will use the trick $\by_{i_s}^{s} = \Enc_{\msf{block}}(\Dec_{\msf{block}}(\bw_{i_s}^s))$ instead of running the local testing algorithm. 
This gives us an \realcc{} with constant rate, constant amortized locality, and constant error tolerance (see Theorem~\ref{thm:FBDCtoARLCC}).

\paragraph{Amortized Locality Preserving Transformation from \realcc{} to \realdc{}}
Next, we show how to preserve constant amortized locality after linearly transforming the \realcc{} into an \realdc{}. As we previously noted the primary challenge is that some of the lowest level blocks may contain relatively few message symbols. We call a block $\bw^s_i$ $c$-bad if it contains at most $cn/2^s$ message symbols. Our key observation is that we can upper bound the total number of message symbols which appear in bad blocks. Let $\Bad{} \subseteq [k]$ be the set of all message symbols which appear in any bad block. We will effectively ignore these message symbols (e.g., by hardcoding them to a fixed symbol like $0$) to obtain a new \realdc{} code for messages of length $k-|\Bad{}|$. While this approach will decrease the final rate of our code the impact will be tolerable as long as we can show that $|\Bad{}|$ is not too large.   

In a bit more detail, suppose that the original code had rate $1 - \Rate{} = k/n$. The rate $1 - \rate{}$ is the proportion of message symbols to codeword symbols, so for any constant $c \leq 1 - \rate{}$, there can only be at most $\frac{\Rate{}}{1- c}$ fraction of $c$-bad blocks (see Lemma~\ref{lem:numGoodBlocks}).
It follows that $|\Bad{}| \leq c \left(\frac{n}{2^s} \right) 2^s \left(\frac{\Rate{}}{1 - c} \right) = cn \left(\frac{\Rate{}}{1 - c} \right)$. Fixing any $c < 1 - \Rate{}$ we have $k'=k-|\Bad{}| = \Omega(k)$.
Given a message $\bx \in \Sigma^{k'}$ we can first preprocess $\bx$ to generate a message $\bx' \in \Sigma^k$ such that $\bx'\left[([k] \setminus \Bad{})\right] = \bx$ and then encode the message $\bx'$. While this lowers the original rate of the code, the new rate will still be constant as long as the original code had constant rate. 

Now we can guarantee that each block $\by^s_i$ at the lowest level either contains no message symbols (in which case we will never need to decode it) or contains {\em at least} $c|\by^s_i| = cn/2^s$ message symbols. 
It follows that if the amortized locality of the \realcc{} is $\alpha_{\realcc{}}$, then the amortized locality of the \realdc{} is $\alpha_{\realdc{}} \leq \alpha_{\realcc{}} / c.$ 
This allows us to obtain an ideal \realdc{} with constant amortized locality, constant rate, and constant error tolerance (see Theorem~\ref{thm:constantARLDC}).

\paragraph{Amortized Locality Less Than $2$ (and Approaching $1$)}   
Lastly, we show that surprisingly, we can instantiate our \realdc{} construction such that the amortized locality is less than $2$ and, for a careful selection of parameters, we can even have the amortized locality approach $1.$
In other words, we can construct an \realdc{} such that the recovery of any sufficiently large block of symbols only takes essentially one queried codeword symbol per message symbol recovered, asymptotically. 

The \realdc{} amortized locality $\alpha_{\realdc{}}$ is determined by the \realcc{} amortized locality $\alpha_{\realcc{}}$ and parameter $c \leq 1 - \rate{\realcc{}}$ in the bad block filtering step, where $\alpha_{\realdc{}} \leq \alpha_{\realcc{}} / c.$ 
Since $c$ is upperbounded by the rate of the \realcc{}, $c < 1 - \rate{\realcc{}}$, we can improve the \realdc{} amortized locality $\alpha_{\realdc{}}$ by lowering the \realcc{} amortized locality $\alpha_{\realcc{}}$ and raising the rate $1 - \rate{\realcc{}}$.
We will show that we can simultaneously achieve (1) amortized locality $\alpha_{\realcc{}}$ approaching $1$, and (2) rate $1 - \rate{\realcc{}}$ less than $2$. 
Observe that when these two conditions are met, then we have constructed an \realdc{} with constant rate, constant error tolerance, and amortized locality less than $2$ (see Corollary~\ref{corr:realdclessthan2}).
Additionally, if we allow the error tolerance to approach zero, use a high rate, sub-constant error-tolerance block code $\calC_{\msf{Block}}$ (constructed in Proposition~\ref{prop:boostedInner}), then the rate $1 - \rate{\realcc{}}$ approaches $1$. In this case, we construct a \realdc{} with rate and amortized locality approaching $1$ (see Corollary~\ref{corr:realdcapproach1}). While the error tolerance will approach $0$ as the locality approaches $1$, for any constant $\epsilon$ we can obtain a  \realdc{} with locality $1+\eps$ and constant error tolerance $\delta = \delta(\eps)$.

\subsection{In-depth Overview}
We provide a detailed overview of the steps we take to achieve our main results, an ideal (constant rate, constant error tolerance, and constant amortized locality) \realdc{} and an \realdc{} with amortized locality approaching $1.$
To start, in the Preliminary Section~\ref{sect:Prelim}, we formally define amortized relaxed locally decodable codes (\realdc{}s) in Definition~\ref{def:realdc}. 
Similarly, we define amortized relaxed locally correctable codes (\realcc{}s) in Definition~\ref{def:realcc}, with the main modification being the recovery of codeword symbols rather than message symbols.

\paragraph{Main Ingredients from Prior Work (Section~\ref{sect:Prelim})}
Next, we recall the main ingredients, vicinity locally testable codes and nesting, used by Cohen and Yankovitz to construct their \rlcc{} \cite{CCC:CohYan24}. 
Vicinity locally testable codes (\VLTC{}s) are the formalization of the locally testable codes with soundness up to relative distance $\delta$ discussed in Section~\ref{subsec:recallingLCC} (see Definition~\ref{def:vltc}).
In more detail, a $(\ell_T,\delta,\rho,\sigma)$-\VLTC{} admits a local testing algorithm $\Test$ which makes at most $\ell_T$ queries to a noisy codeword $\bw$. The guarantee is that if $\reldist(\bw,\by) \leq \delta$ then $\Test{}^{\bw}$ returns $\bot$ with probability at least $\rho \times \reldist(\bw,\by) - \sigma$ --- if $\reldist(\bw,\by) = 0$ then $\Test{}^{\bw}$ should always accept.
Crucially, Cohen and Yankovitz construct bipartite expander graph-based \VLTC{}s.
Given a bipartite expander $G = (L,R,E)$, one can construct a code such that (1) the left vertices $L$ represent the codeword symbols and the right vertices $R$ represent the parity checks of the code (each left vertex $l \in L$ is an index of code $\calC$ of codeword length $|L|$ and each right vertex $r \in R$ enforces $\sum_{(l,r) \in E} \by[l] = 0$), and (2) there exists a tester \Test{} that picks a right vertex $r \in R$ at random, checks $\sum_{(l,r) \in E} \by[l] = 0$ by querying each codeword symbol $\by[l]$ such that $(l,r) \in E,$ and accepts if and only if the parity check is consistent (see Definition~\ref{def:expander}). 
Intuitively, by the expansion property of the bipartite expander, any error in the codeword will be a part of many parity checks (see Lemma~\ref{lem:vtester}).

We reintroduce the nesting operation $\divideontimes$ first devised by Kumar and Mon \cite{STOC:KumMon24} and used by Cohen and Yankovitz to (recursively) nest their \VLTC{}s into a \rlcc{} \cite{CCC:CohYan24}. Nesting is the formal operation which allows each block on each level to be locally testable as discussed in Section~\ref{subsec:recallingLCC}. 
Suppose $\calC$ and $\calC'$ are codes with codeword length $n$ and $n/2$, respectively, for even $n \in \N.$ 
Then, any codeword $\by$ in the nested code $\calC \divideontimes \calC'$ (1) is a codeword of $\calC$, and (2) its halves, $\by[1:n/2]$ and $\by[n/2 + 1:n]$ are also codewords of $\calC'.$
More formally, define the (linear) nested code of codes $\calC$ and $\calC'$ as $\calC \divideontimes \calC' = \calC \divideontimes (\calC')^2$ (see Definition~\ref{def:nested} and Definition~\ref{def:recNested} for the recursively applied nesting definition). 
In general, when nesting codes $\calC$ and $\calC'$, with rates $1 - \Rate{}$ and $1 - \Rate{}'$ respectively, nested code $\calC \divideontimes \calC'$ has rate at least $1 - \Rate{} - \Rate{}'$
(Lemma~\ref{lem:nestingParams} and Corollary~\ref{corr:nestedRate}).
Notably, the rate of nested code $\calC \divideontimes \calC'$ has a nice interpretation when $(n,k)$ code $\calC$ and $(n/2, k')$ code $\calC'$ are expander codes. 
Informally, $\calC$ has $\Rate{}n$ parities/constraints on its $n$ codeword symbols if it has rate $1 - \Rate{}$.
Similarly $\calC'$ has $\Rate{}' \times (n/ 2)$ constraints on its $n/2$ codeword symbols if it has rate $1 -\Rate{}'$.
Then, the nested code $\calC \divideontimes \calC'$ with codeword length $n$ has $\Rate{} + \Rate{}'$ constraints where there are $\Rate{}n$ constraints on all $n$ codeword symbols $\by$, $R' \times (n/2)$ constraints on the first $n/2$ codeword symbols $\by[1:n]$, and $\Rate{}' \times (n/2)$ constraints on the last $n/2$ codeword symbols $\by[n+1:2n]$; subsequently, there are a total of $\Rate{} + \Rate{}'$ constraints, implying nested code $\calC \divideontimes \calC'$ has rate at least $1 - \Rate{} - \Rate{}'.$

\input{tikz/construction}

\paragraph{Constructing \FBDC{}s (Section~\ref{sect:FBDC})} 
To obtain an amortized version of the \rlcc{} construction of Cohen and Yankovitz, we introduce fuzzy block decoding codes (\FBDC{}), which isolate the features of their construction that enable amortization.
\FBDC{}s are codes whose codewords have a predefined block structure and a `block decoder' that either recovers a requested block of the codeword with some allowable fraction of error or outputs $\bot$ on detection of an error.
More formally, a $(\ell,B,\delta,\delta_B,\eps)$-\FBDC{} is a code with $\delta$ distance and a decoder \BlockDec{} that, given $\ell$ queries to the (possibly corrupted) codeword $\by = \bw_1\circ \dots \bw_{|\by| / B}$ such that each $|\bw_p| = B$ and input $j \in [|\by| / B],$  recovers block $\bw_j$ with at most a $\delta_B$  fraction of errors or outputs $\bot$. 
In Proposition~\ref{prop:recursiveExpander}, essentially following the construction of Cohen and Yankovitz \cite{CCC:CohYan24}, we show how to construct \FBDC{}s by recursively nesting expander codes that double in codeword length at each nesting step.
To start, let $\{G_i = (L_i,R_i,E_i)\}_{i \in \N}$ be a sequence of bipartite expander graphs with a corresponding sequence of expander codes (that are \VLTC{}s) $\{(n_i,k_i) \text{ code } \calC_{i}\}_{i \in \N}$ where we have $n_{i + 1}  = 2n_i$ for all $i \in \N.$ 
For any nesting parameter $s \in \N$, we define our code as $\calC_{\FBDC{}} = \Sigma^{2^{s}\times  n_1} \divideontimes \calC_s \divideontimes \calC_{s - 1} \divideontimes \dots \divideontimes \calC_1$, pictured in Figure~\ref{fig:constr}. 
Then, $\calC_{\FBDC{}}$ is a $(n = 2^s \times n_1, k)$ code with rate $1 - \rate{\FBDC{}} \geq 1 - \sum_{i \leq s}(1 - \frac{k_i}{n_i})$, where each codeword $\by$ can be partitioned into blocks $\by = \by_1^j \circ \dots \circ \by_{n / 2^j}^j$ for any $j \leq s$ such that $\by_i^j$ is a $(\ell_T,\delta_T,\rho,\sigma)$-\VLTC{}. 
We construct a \FBDC{} decoder \BlockDec{} to show that $\calC_{\FBDC{}}$ is a $(\ell = s\ell_T + n_1/2, B = n_1/2, \delta' = \delta_T / 2, \delta_B = 2\delta_T,\eps = O(\rho\delta - \sigma))$-\FBDC{}. 
Note the absence of the top-most \VLTC{} and that the lowest level block size is half of the smallest \VLTC{} block (which we address using top level amplification in Lemma~\ref{lem:topLevelAmp} and block coding in Theorem~\ref{thm:FBDCtoARLCC} respectively).
Intuitively, on a word $\bw = \bw^j_{1} \circ \dots \bw_{2^j}^j$ such that $\reldist(\bw,\by) \leq \delta'$, we can ensure any block $\bw^{s}_{i_s}$ has at most $\delta'$ errors by applying the \VLTC{} test on each containing block $\bw^1_{i_1},\dots,\bw^{s-1}_{i_{s-1}},\bw^s_{i_s}$.
More specifically, suppose $\BlockDec{}^{\bw}(i)$ is tasked with recovering the $i$'th block of $\bw$, which we denote as $\bw_{i_{s + 1}}^{s+1} = \bw[(s+1)B + 1: (s+2)B]$.
Let $i_1,\dots,i_s$ be the block indices such that $\bw_{i_j+1}^{j+1}$ contains block $\bw_{i_{j }}^{j}$  for all $j \in [s].$
\BlockDec{} runs the \VLTC{} tester $\Test{}^{\bw^j_{i_j}}$ on the block $\bw^j_{i_j}$ for each $j \in [s] $. If $\Test{}^{\bw^j_{i_j}} = \bot$ for any $1 \leq j \leq s$ (any local test fails) then the decoder $\BlockDec^{\bw}(i)$ will output $\bot$. Otherwise, if all local tests pass, we output $\bw_{i_{s + 1}}^{s+1}$.
Observe that if no soundness errors occur at any level, then $\bw_{i_s}^s$ only has at most a $\delta'$ fraction of errors, i.e., $\reldist(\by_{i_s}^s,\bw_{i_s}^s) \leq \delta'$. 
This implies that the block $\bw_{i_{s + 1}}^{s+1}$ only has at most a $2\delta'$ fraction of errors since $2|\bw_{i_{s + 1}}^{s+1} | = |\bw_{i_s}^s|$, as desired. 
Thus, the soundness error of the \FBDC{} is at most the probability that for all levels $j \in [s],$ the \VLTC{} tester $\Test{}$ does not output $\bot.$
Using the same clever analysis of Cohen and Yankovitz \cite{CCC:CohYan24}, we can avoid a union bound over the number of levels $s$ by showing that at some level $j \in [s]$, there must be a block $\bw_{i_j}^j$ with a bounded fraction of errors, $\reldist(\bw_{i_j}^j,\by_{i_j}^j) \in [\delta/2,\delta]$. 

We emphasize that a key difference between our soundness analysis is that Cohen and Yankovitz require an upperbound on the block size in order to prove \rldc{} soundness.
Cohen and Yankovitz sufficiently set the lowest level block size to be less than $1 /\delta'$ so that any block $\bw_i^s$ with at least one error has $\reldist(\bw_i^s,\by_i^s) > \delta'$.
Intuitively, they argue that this implies that there exists some level $j \in [s]$ where the error rate is bounded $\reldist(\by^j_i,\bw_i^j) \in [\delta'/2,\delta],$ and hence, soundness is reduced to the soundness of the \VLTC{} on level $j.$
In our case, the precondition, that $\reldist(\bw_i^s,\by_i^s) > \delta'$ is already baked into the \FBDC{} definition, i.e., soundness error only occurs if we return a block with error fraction higher than $\delta'.$
This crucially allows us to set arbitrarily large block sizes so that we can guarantee constant amortized locality.

Note that we can get the \FBDC{} soundness error down by an additional power of $g \in \N$ by doing a canonical amplification step of running any \VLTC{} tester $\Test{}$ an additional factor of $O(g)$ times and outputting $\bot$ if and only if any $\Test$ outputs $\bot.$
Looking forward a bit, in contrast to traditional locality, due to our consideration of {\em amortized} locality, we can account for the $O(g)$ factor increase in the number of queries in our construction by increasing the minimal number of symbols recovered $\kappa$ by the same $O(g)$ factor to keep the amortized locality constant or approaching $1$.

We now address the absent top level \VLTC{}. 
Recall from the discussion in Section~\ref{subsec:recallingLCC} that the \rlcc{} error tolerance is necessarily sub-constant from nesting the \VLTC{}s $s = O(\log n)$ times if the overall rate is constant.
Our \FBDC{} construction will suffer the same sub-constant top level error tolerance $\delta'$ since we aim make our \FBDC{} block sizes as small as possible while maintaining low amortized locality in the \realdc{}. 
Hence, we introduce Lemma~\ref{lem:topLevelAmp}, which boosts the top level error tolerance of a \FBDC{} by applying a top level \VLTC{} with good error tolerance error tolerance (also pictured in Figure~\ref{fig:constr}).
Note that we can afford to apply a high error tolerance \VLTC{} because we are only nesting it once.
More specifically, given a $(n,k_{\FBDC{}})$ code $\calC_{\FBDC{}} = (\cdot,\BlockDec{})$ that is a $(\ell_{\FBDC{}}, \delta',\delta_B,\eps_{\FBDC{}})$-\FBDC{} and a $(2n,k_{\VLTC{}})$ code $\calC_{\VLTC{}} = (\cdot,\Test)$ that is a $(\ell_T,\delta,\rho,\sigma)$-\VLTC{}, the $(n,k)$ code $\calC_{\msf{amp}} = \calC_{\VLTC{}} \divideontimes \calC_{\FBDC{}}$ is a $(\ell = \ell_{\FBDC{}} + O(\delta/\delta') \times \ell_T, \delta,\delta_B,\eps = O(\eps_{\FBDC{}}))$-\FBDC{}. 
The new block decoder $\BlockDec{}_{\msf{amp}}$ for code $\calC_{\msf{amp}}$ will first attempt to detect error at the top level using the tester \Test{} $O(\delta/\delta')$ times; if all tests pass, then the prior block decoder \BlockDec{} is used. 
The idea is that for any corrupted word $\bw$ such that $\reldist(\bw,\by) \in (\delta'/2,\delta]$, the errors will be caught by at least one $\Test$ invocation with constant probability since the number of times we repeat the test is proportional to the error tolerance ratio $2\delta / \delta'.$
Then, if all tests pass, this implies that the error rate is low, $\reldist(\bw,\by) \leq \delta'$ (additionally implying that $\reldist(\bw^1_1,\by^1_1), \reldist(\bw^1_2,\by^1_2) \leq \delta'$), so we can defer to using the original $\BlockDec{}^{\bw_{i_1}^1}$ to recover block $\bw^s_{i_s}$.
By applying a top level amplification of the error tolerance, we can construct a $(\ell,B,\delta,\delta_B,\eps)$-\FBDC{} where the top level error tolerance $\delta$ is constant while the block level error tolerance $\delta_B$ is sub-constant.

\paragraph{Converting \FBDC{}s into \realcc{}s (Section~\ref{sect:FBDC})}
We show how to transform the \FBDC{} into an \realcc{} in Theorem~\ref{thm:FBDCtoARLCC}.
So far, we have constructed a code $\calC_{\FBDC}$ with block decoder $\BlockDec{}$ that is a $(\ell,B,\delta,\delta_B,\eps)$-\FBDC{}. 
We now wish to construct a code $\calC_{\realcc{}}$ with relaxed local decoder $\Dec{}$ that is a $(\alpha,\kappa,\delta,\eps)$-\realcc{}.
Observe that the main difference between a \FBDC{} and an \realcc{} is that the recovered symbols by the \FBDC{} decoder, say $\bw^{s}_{i_s},$ can have up to a $\delta_B$ fraction of errors, $\reldist(\bw^{s}_{i_s},\by^s_{i_s}) \leq \delta_B$, while the recovered symbols by the \realcc{} decoder are error-free (if both decoders do not output $\bot$).
Our insight is that if we make each block $\by^{s}_{i_s}$ a codeword of an error correcting code tolerating $\delta_B$ fraction of errors, then we can recover error-free block $\by^{s}_{i_s}$ from corrupted block $\bw^{s}_{i_s}$, thus converting the $(\ell,B,\delta,\delta_B,\eps))$-\FBDC{} into a $(\alpha = O(\ell / B),\kappa = B,\delta,\eps)$-\realcc{}. 
Specifically, we apply an additional nesting of a $(B,k_{\msf{block}})$ code $\calC_{\msf{Block}} = (\Enc,\Dec)$ tolerable to $\delta_B$ errors, $\calC_{\realcc{}} = \calC_{\FBDC{}} \divideontimes \calC_{\msf{Block}}$.
Then, suppose the relaxed local decoder $\Dec$ for code $\calC_{\realcc{}}$ is given as input the range $[L,R] \subseteq [n]$, where we suppose $R - L + 1 \geq B$.
There exists a minimal block index range $i_{j},\dots, i_{j}+t - 1$ such that $[L,R] \subseteq \{i_{j }B + 1,\dots, (i_{j}+ t)B\}$. 
Hence, we can run the \FBDC{} decoder $\BlockDec{}^{\bw}$ $t$ times on input indices $i_j,\dots,$ and $ i_j+ t$,  to obtain $\bw^s_{i_j}, \dots, \bw^s_{i_j + t - 1}.$
For each $q \in [t],$ we can recover $\by^s_{i_j + q - 1}$ by decoding and re-encoding $\bw^s_{i_j + q - 1}$, i.e., $\Enc(\Dec(\bw^s_{i_j + q - 1})) =  \by^s_{i_j + q - 1}$ since $\bw^s_{i_j + q - 1}$ is guaranteed to have at most a $\delta_B$ fraction of errors (if no soundness error occurs). 
Completeness and soundness follows from the \FBDC{} correctness, and we attain amortized locality $\alpha = O(\ell /B)$ since $t$ invocations of $\BlockDec{}^{\bw}$ makes $\ell t$ queries and we recover at least $B \times (t- 1)$ requested codeword symbols.
Looking forward, when we instantiate the \realcc{}, we will set the block size equal to the number of queries made, $B = \ell$, to achieve constant amortized locality (along with constant rate and constant error tolerance).
We emphasize that we are able to increase the block size $B$ arbitrarily since we are no longer reliant on a small block size for soundness, in contrast to the prior \rlcc{} construction of Cohen and Yankovitz \cite{CCC:CohYan24}.
To further improve on the locality, we can set the minimum recovery size $\kappa$ to be asymptotically larger than the block size $B,$ i.e., $\kappa = hB$ for some $h = \omega(1),$ to construct an \realcc{} with amortized locality approaching $1$ (along with constant rate and constant error tolerance).

\paragraph{\realcc{} to \realdc{} by Bad Block Filtering (Section~\ref{sect:LCCtoLDC})}
As previously discussed in Section~\ref{subsec:challenges} challenge $2$, while the conversion from traditional linear \LCC{}s to linear \LDC{}s is straightforward and locality-preserving via a linear transformation, the same transformation used for the conversion from linear \realcc{}s to a linear \realdc{}s need not preserve the {\em amortized} locality. 
Succinctly, recall that some blocks $\by_j^s$ may contain relatively few message symbols after the linear transformation is applied onto the \realcc{}, so we no longer have a guarantee on the minimum number of blocks recovered. 
Hence, the amortized locality is not preserved a priori. 
To address this, we implement the solution outlined in Section~\ref{subsec:contributions}, where we effectively ignore any message symbols that appear within a block with a low proportion of message symbols. 
First, we define $c$-good (resp. bad) blocks as blocks $\by_j^s$ which contain either zero or at least (resp. less than) a $c$ fraction of message symbols (Definition~\ref{def:goodBlock}). 
Our objective is to construct a code whose codeword blocks are all $c$-good, which we denote as a $c$-message smooth code (Definition~\ref{def:smoothCode}).
For a $(n,k)$ code $\calC$, we show for a block size $B$ and any fraction $c \leq k/n$ that code $\calC$ contains at least $\frac{k - cn}{(1- c) B}$ $c$-good blocks by maximizing the number of $c$-bad blocks code $\calC$ could possibly contain (Lemma~\ref{lem:numGoodBlocks}).
This implies that by ignoring all $c$-bad blocks, e.g., by fixing their message symbols to $0$'s, the number of message symbols remaining is at least $\frac{k - cn}{(1- c) B} \times B = \frac{k - cn}{1 - c}$, and the rate is subsequently at least $\frac{1 - R - c}{1 - c}$ (Theorem~\ref{thm:SmoothifyingNestedCodes}).
Thus, if we apply Theorem~\ref{thm:SmoothifyingNestedCodes} to our \realcc{} construction after linear transformation, we can guarantee each block recovered by the decoder to have at least a $c < 1 - \Rate{\realcc{}}$ fraction of message symbols, raising the amortized locality by a $1/c$ factor and lowering the rate accordingly in the \realdc{}(Corollary~\ref{corr:arLCCToArLDC}).

\paragraph{Explicit Instantiation of an Ideal \realdc{} (Section~\ref{sect:AllTogether})}
We are ready to explicitly instantiate our \realdc{} constructions now that we have shown how to construct a \FBDC{} from recursively-nested expanders (Proposition~\ref{prop:recursiveExpander} and Lemma~\ref{lem:topLevelAmp}), an \realcc{} from a \FBDC{} (Theorem~\ref{thm:FBDCtoARLCC}), and an \realdc{} from an \realcc{} (Corollary~\ref{corr:arLCCToArLDC}).
First, we instantiate our \realdc{} construction with the best known explicit expanders (Theorem~\ref{thm:explicitExpander}) and demonstrate that our construction achieves constant amortized locality (Theorem~\ref{thm:constantARLDC}).
Our key consideration will be setting block size $B$ to be large enough to asymptotically match the number of queries.
Following the parameters set by Cohen and Yankovitz in their explicit \rlcc{} construction \cite{CCC:CohYan24}, we can construct a $(\ell ,B,\delta = \Omega(1),\delta_B,\eps = O(1))$-\FBDC{} with an equivalent number of queries $\ell = (\log n)^2 \times 2^{O((\log\log \log n)^3)}$, block error tolerance $\delta_B = \Omega\left(\frac{1}{\log n\log\log n}\right)$, and rate $1 - \Rate{}' \geq 1 - \delta \times 2^{O((\log \log 1/\delta)^3)} - o(1)$.
Thus, barring any soundness amplification, we can set the block size to be $B = (\log n)^2\times 2^{(\log\log \log n)^4}$, apply a constant rate, constant error tolerance code (e.g., a Justesen Code \cite{ITIT:Jus72}) in the \FBDC{} to \realcc{} step, to get a $(\alpha' = O(1),\kappa = (\log n)^2\times 2^{(\log\log \log n)^4},\delta = \Omega(1), \eps = O(1))$-\realcc{} with constant rate $1 - \rate{\realcc{}} \geq 1 - \rate{}' - \rate{\msf{Block}}.$
Note that we can further state that the amortized locality of the \realcc{} $\alpha'$ approaches $1$ since the block size $B$ is asymptotic to the number of \FBDC{} queries $\ell$ (technically, we also have to raise the number of blocks recovered, $\kappa = hB$ for $h \in \omega(1)$, as previously discussed). 
Lastly, in the \realcc{} to \realdc{} step, we can pick $c$ proportional to the rate $1 - \Rate{}'$, say $c = (1 - \Rate{}')/ 2$ to obtain a $( \alpha'\times (2/(1 - \Rate{}')) = O(1),\kappa, \delta,\eps)$-\realdc{} with constant rate $1 - \Rate{} \geq \frac{(1 - \rate{\realcc{}}) / 2}{\rate{\realcc{}}/2} \geq  \Omega(1).$

\paragraph{Lowering the Amortized Locality (Section~\ref{sect:constrApproach1})}
We further show that our ideal \realdc{} construction can be modified to have a more precisely stated amortized locality.
In particular, we show that either (1) we can obtain amortized locality $< 2$ while keeping rate and error tolerance constant, or (2) we can have amortized locality and rate approach $1$, if the error tolerance is allowed to be sub-constant.
First, in the prior ideal instantiation, we start with the  intermediately constructed $(\ell,B,\delta,\delta_B,\eps)$-\FBDC{} with number of queries $\ell = (\log n)^2 \times 2^{O((\log\log \log n)^3)}$, block error tolerance $\delta_B = \Omega\left(1/\log n\log\log n\right)$, and rate $1 - \Rate{}' \geq 1 - \delta \times 2^{O((\log \log 1/\delta)^3)} - o(1)$.
By the note in the prior discussion, we can set block size $B = (\log n)^2\times 2^{(\log\log \log n)^4}$ to attain \realcc{} amortized locality approaching $1.$
Additionally, we observe that the choice of the constant rate, constant error tolerant Justesen code for our block code is excessive in the \FBDC{} to \realcc{} step since we only need a block code $\calC_{\msf{Block}}$ tolerant to a $\delta_B = \Omega\left(1/\log n\log\log n\right)$ fraction of errors.
Instead of using the Justesen code, we show how to efficiently construct code $\calC_{\msf{Block}}$ with error tolerance $\Omega\left(\frac {\log\log n} {\log n}\right)$ and rate  $1 - \rate{\msf{Block}} \geq 1 - {1}/ {\log\log n} - {(\log \log n)^3}/{\log n} = 1 - o(1)$ (see Proposition~\ref{prop:boostedInner}).
Penultimately, with a sufficiently high block size $B$ and a more appropriate block code $\calC_{\msf{Block}}$, we construct a $(\alpha' = 1 + o(1), \kappa, \delta, \eps = O(1))$-\realcc{} code $\calC_{\realcc{}}$ with rate $1 - \rate{\realcc{}} \geq 1 - \delta \times 2^{O((\log \log 1/\delta)^3)} - o(1)$. 
Finally, we choose parameter $c \leq 1 - \rate{\realcc{}}$ to obtain a $(\alpha = (1+ o(1)) / c, \kappa, \delta, \eps)$-\realdc{} with rate $1 - \rate{\realdc{}} \geq \frac{1 - \rate{\realcc{}} - c}{1 - c}$ (see Theorem~\ref{thm:approachingOneARLDC}).
Observe that there exists constant $\delta > 0$ (dependent on the constant factor in the $2^{O(\cdot)}$ term) such that $1 - \rate{\realcc{}} > 1/2$. 
Therefore, there is a $c$ value such that the amortized locality $\alpha = \alpha_{\realcc{}} / c < 2$ and rate $1 - R_{\realdc{}} \geq \Omega(1)$ (see Corollary~\ref{corr:realdclessthan2}). 
Moreover, we can achieve amortized locality and rate approaching $1$ by allowing the error tolerance to become sub-constant. 
In more detail, we set $\delta \in o(1)$ such that the rate $1 - \rate{\realcc{}} \geq 1 - \delta \times 2^{O((\log \log 1/\delta)^3)} - o(1)$ approaches $1$, say $1 - \rate{\realcc{}} \geq  1 - \frac{1}{\log n}$. 
Then, by setting $c$ to approach $1 - \rate{\realcc{}}$, say $c = 1 - \frac{1}{\log\log n} < 1 - \rate{\realcc{}}$, we can attain $\realdc{}$ amortized locality $\alpha = (1 +o(1)) / c = 1 + o(1)$, while the rate is at least $\frac{ \frac{1}{\log \log n} - \frac 1 {\log n}}{\frac{1}{\log\log n}} = 1 - \frac{\log\log n}{\log n}$ (see Corollary~\ref{corr:realdcapproach1}).

We conclude the technical overview by presenting our three main results in full.

\input{theorems/constantARLDC}
\input{theorems/lessthan2}
\input{theorems/approach1}

%% file: tikz/binaryStructure.tex
\newcommand{\blockdiagram}[3]{%
\begin{scope}[xshift=#1]

    \draw[line width=3pt] (0,0) rectangle (6,0.8);
    \node at (3,0.4) {#2};
    
    \draw (3.7,1) -- (3.7,0.6);
    \node at (3.7,1.2) {$i$};

    \draw[dashed](3.7,1) -- (3.7,-4.6);
    
    \begin{scope}[yshift=-1.4cm]
        \draw (0,0) rectangle (3,0.65);
        \draw[line width=3pt] (3,0) rectangle (6,0.65);
        \node at (1.5,0.32) {$#3_1^1$};
        \node at (4.5,0.32) {$#3_2^1$};
    \end{scope}
    
    \begin{scope}[yshift=-2.8cm]
        \foreach \j in {0,1,2,3} {
            \draw ({1.5*\j},0) rectangle ({1.5*(\j+1)},0.65);
        }
        \draw[line width=3pt] (3,0) rectangle (4.5,0.65);
        \node at (0.75,0.32) {$#3_1^2$};
        \node at (2.25,0.32) {$#3_2^2$};
        \node at (3.75,0.32) {$#3_3^2$};
        \node at (5.25,0.32) {$#3_4^2$};
    \end{scope}
    
    \node at (3,-0.4) {\rotatebox{90}{$=$}};
    \node at (3,-1.8) {\rotatebox{90}{$=$}};
    \node at (3,-3.1) {\rotatebox{90}{$=$}};
    \node at (2.95,-3.6) {$\vdots$};
    \node at (3,-4.2) {\rotatebox{90}{$=$}};
    
    \begin{scope}[yshift=-5.0cm]
        \draw (0,0) rectangle (6,0.55);
    
        \foreach \x in {0.35,0.7,1.05,1.4} {
            \draw (\x,0) -- (\x,0.55);
        }
    
        \draw (5.55,0) -- (5.55,0.55);
        
        \draw[line width=3pt] (3.45,0) -- (3.45,0.55);
        \draw[line width=3pt] (4.1,0) -- (4.1,0.55);
        \draw[line width=3pt] (3.4,0) -- (4.15,0);
        \draw[line width=3pt] (3.4,0.55) -- (4.15,0.55);
    
        \node[font=\small] at (0.18,0.28) {$#3_1^s$};
        \node at (3,0.28) {$\cdots$};
        \node[font=\small] at (3.78,0.28) {$#3_{i_s}^{s}$};
        \node[font=\small] at (5.78,0.28) {$#3_{2^s}^{s}$};
    \end{scope}

\end{scope}
}

\begin{figure}
    \centering
\resizebox{0.5\textwidth}{!}{\begin{tikzpicture}[
    box/.style={draw, minimum height=0.65cm},
    every node/.style={font=\large}
]

\blockdiagram{-8cm}{Codeword $\by$}{\by}

\draw[->, line width=1pt] (-1.45,-2.35) -- (-0.35,-2.35)
    node[midway, above, font=\normalsize] {$\delta n$ errors};
\blockdiagram{0cm}{Corrupted word $\bw$}{\bw}
\end{tikzpicture}}
    \caption{Codeword $\by = (\by^{1}_1 \circ \by_2^1) = (\by_1^2 \circ \by_2^2 \circ \by_3^2 \circ \by_4^2) = \dots = (\by_1^s \circ \dots \circ \by_{2^s}^s)$ is constructed such that each intermediate $\by_{i_j}^j$ block is locally testable.
    Then, the decoder $\Dec^{\bw}$ with oracle access to word $\bw = (\bw^{1}_1 \circ \bw_2^1) = (\bw_1^2 \circ \bw_2^2 \circ \bw_3^2 \circ \bw_4^2) = \dots = (\bw_1^s \circ \dots \circ \bw_{2^s}^s)$ ($ \leq \delta n$ errors from codeword $\by$) and input index $i \in [n]$ runs the local tests on the blocks of word $\bw$ containing index $i.$
    That is, for the sequence $(i_0,i_1,\dots,i_s)$ such that $\by[i] \in \by^0_{i_0} \supset \by^1_{i_1} \supset \by_{i_2}^2 \supset \dots \supset \by_{i_s}^s,$ run the local tests on blocks $\bw^0_{i_0},\bw^1_{i_1}, \bw_{i_2}^2, \dots, \bw_{i_s}^s$, and output $\bw[i]$ if and only if all tests pass.}
    \label{fig:blockstruct}
\end{figure}
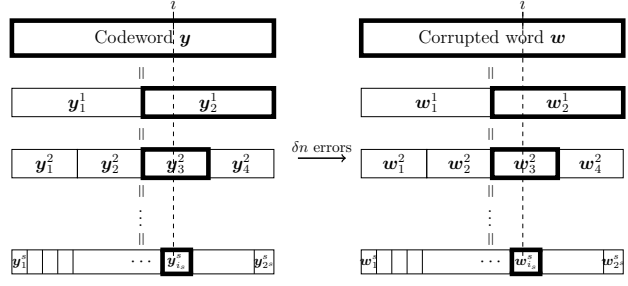

%% file: tikz/construction.tex
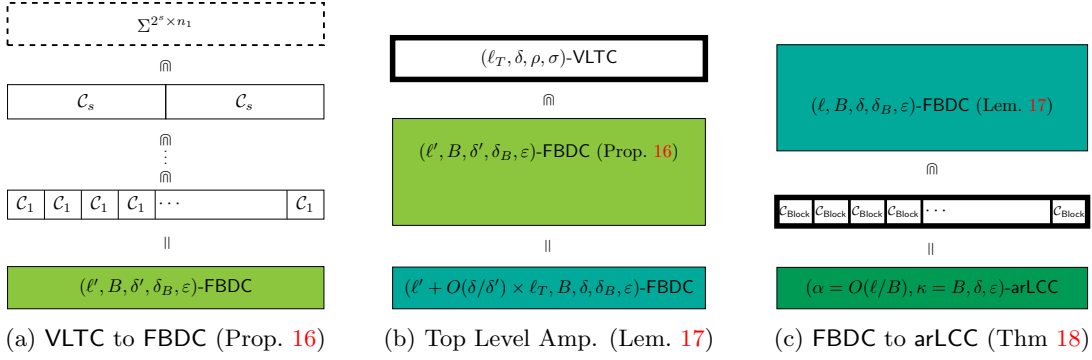
\begin{figure}
    \centering
\begin{subfigure}{.3\textwidth}\label{subfig:vltctoFBDC}\centering
\resizebox{.85\textwidth}{!}{\begin{tikzpicture}
    \begin{scope}[xshift=0cm]

    \draw[line width=1pt, dashed] (0,0) rectangle (6,0.8);
    \node at (3,0.4) {$\Sigma^{2^s \times n_1}$};
    
    \begin{scope}[yshift=-1.4cm]
        \draw (0,0) rectangle (3,0.65);
        \draw (3,0) rectangle (6,0.65);
        \node at (1.5,0.32) {$\calC_s$};
        \node at (4.5,0.32) {$\calC_s$};
    \end{scope}

    \node at (3,-0.4) {$\divideontimes$};
    \node at (3,-1.8) {$\divideontimes$};
    \node at (3,-2.5) {$\divideontimes$};
    \node at (3,-2.05) {\scriptsize$\vdots$};
    \node at (3.05,-3.75) {\rotatebox{90}{$=$}};
    
    \begin{scope}[yshift=-3.3cm]
        \draw (0,0) rectangle (6,0.55);
    
        \foreach \x in {1,2,3,4} {
            \draw (\x * .7,0) -- (\x * .7,0.55);
            
        }
        \foreach \x in {0,1,2,3} {
            \node[] at (\x*0.7+0.35,0.28) {$\calC_1$};
        }
    
        \draw (5.3,0) -- (5.3,0.55);

        \node at (3.1,0.28) {$\cdots$};
        \node[] at (5.65,0.28) {$\calC_1{}$};
    \end{scope}
    \begin{scope}[yshift=-3.4cm]
        \draw[fill=LimeGreen] (0,-1.6) rectangle (6,-.8);
        \node[] at (3,-1.2) {$(\ell' ,B,\delta',\delta_B,\eps)$-\FBDC{} };
    \end{scope}
\end{scope}
    \end{tikzpicture}}
\caption{\VLTC{} to \FBDC{} (Prop.~\ref{prop:recursiveExpander})}
\end{subfigure}
\begin{subfigure}{.3\textwidth}\label{subfig:toplevelamp}\centering
             \resizebox{.85\textwidth}{!}{\begin{tikzpicture}
    \begin{scope}[xshift=0cm]

    \node at (3,-.4) {$\divideontimes$};
    \node at (3,-3.2) {\rotatebox{90}{$=$}};
    
    \draw[line width=3pt] (0,0) rectangle (6,0.8);
    \node at (3,0.4) {$(\ell_T,\delta,\rho,\sigma)$-$\VLTC{}$};
    
    \begin{scope}[yshift=-1.4cm]
        \draw[fill=LimeGreen] (0,-1.4) rectangle (6,0.65);
        \node at (3,0) {$(\ell',B,\delta',\delta_B,\eps)$-\FBDC{} (Prop.~\ref{prop:recursiveExpander})};
    \end{scope}
    
    \begin{scope}[yshift=-2.8cm]
        \draw[fill=JungleGreen] (0,-1.6) rectangle (6,-.8);
        \node[] at (3,-1.2) {$(\ell' + O(\delta/\delta') \times \ell_T,B,\delta,\delta_B,\eps)$-\FBDC{} };
    \end{scope}

\end{scope}
    \end{tikzpicture}}
    \caption{Top Level Amp. (Lem.~\ref{lem:topLevelAmp})}
\end{subfigure}
\begin{subfigure}{.3\textwidth}\label{subfig:fbdctoarlcc}\centering
\resizebox{.85\textwidth}{!}{\begin{tikzpicture}
    \begin{scope}[xshift=0cm]

    \node at (3,-2.4) {$\divideontimes$};
    \node at (3,-3.95) {\rotatebox{90}{$=$}};

    \begin{scope}[yshift=0cm]
        \draw[fill=JungleGreen] (0,0) rectangle (6,-2.05);
        \node[] at (3,-1.2) {$(\ell,B,\delta,\delta_B,\eps)$-\FBDC{} (Lem.~\ref{lem:topLevelAmp})};
    \end{scope}

        \begin{scope}[yshift=-3.5cm]
        \draw[line width=3pt] (0,0) rectangle (6,0.55);
    
        \foreach \x in {1,2,3,4} {
            \draw[line width=1.5pt] (\x * .7,0) -- (\x * .7,0.55);
            
        }
        \foreach \x in {0,1,2,3} {
            \node[font=\scriptsize] at (\x*0.7+0.35,0.28) {$\calC_{\msf{Block}}$};
        }
    
        \draw[line width=1.5pt] (5.3,0) -- (5.3,0.55);

        \node at (3.1,0.28) {$\cdots$};
        \node[font=\scriptsize] at (5.65,0.28) {$\calC_{\msf{Block}}$};
    \end{scope}
    \begin{scope}[yshift=-3.5cm]
        \draw[fill=ForestGreen] (0,-1.6) rectangle (6,-.8);
        \node[] at (3,-1.2) {$(\alpha = O(\ell / B), \kappa = B, \delta,\eps)$-\realcc{} };
    \end{scope}
\end{scope}
    \end{tikzpicture}}
    \caption{\FBDC{} to \realcc{} (Thm~\ref{thm:FBDCtoARLCC})}
\end{subfigure}
    \caption{Subfigures a, b, and c describe how we construct an \realcc{} from a family of expander graphs.
    First, in Proposition~\ref{prop:recursiveExpander}, we construct a $(\ell',B,\delta',\delta_B,\eps)$-\FBDC{} by nesting expander code-based \VLTC{} family $\{\calC_i\}_{i \in \N}$ (Lemma~\ref{lem:vtester}). We do not initially nest the top-level \VLTC{} so that we may perform distance amplification (Lemma~\ref{lem:topLevelAmp}). Intuitively, to catch any intermediate error fraction between $\delta'$ and $\delta$, we need to repeat the top-level \VLTC{} tester $O(\delta/\delta')$ times. We then apply block-level nesting of an error correcting code $\calC_{\msf{Block}}$ tolerant to $\delta_B$ errors to obtain a $\realcc{}$ in Theorem~\ref{thm:FBDCtoARLCC}.}
    \label{fig:constr}
\end{figure}

%% file: theorems/constantARLDC.tex
\begin{restatable*}[Constant Rate, Error-Tolerance, Amortized Locality $\realdc{}$]{theorem}{constantARLDC} \label{thm:constantARLDC}
For any $n = 2^r$, $r \in \N$, and $\eps < 1$, there is an (explicit) $(n,k)$ code $\calC$ that is a $(\alpha = O(1),\kappa = \log\left(1/\eps\right) \times (\log n)^2 \times 2^{(\log\log\log n)^4},\delta = \Omega(1),\eps)$-$\realdc{}$ with rate $1 - \Rate{\realdc{}} \geq \Omega(1).$
\end{restatable*}

%% file: theorems/lessthan2.tex
\begin{restatable*}[\realdc{} with Amortized Locality $< 2$]{corollary}{lessthan} 
\label{corr:realdclessthan2}
    Let $\Rate{}(\delta) = \delta \times 2^{c_0 \log^3(\log 1/\delta)}$ for some universal constant $c_0 \in \N.$ For any $n = 2^r$, $r \in \N,$ constants $ \eps < 1,$ $\nu \in (0,2/3),$ and $\delta \in (0,1)$ such that $\Rate{}(\delta) < 1 - \frac{3(1 + \nu)}{5},$  there exists an (explicit) $(n,k)$ code $\calC$ that is a $\left(\alpha = 5/3, \kappa(\nu) , \delta, \eps\right)$-$\realdc{}$, where $\kappa(\nu) = \frac{ \log(1/\eps)\times(\log n)^2\times 2^{(\log\log\log n)^4}}{\nu}$ and the rate $1 - \rate{\realdc{}} \geq \Omega(1)$. 
\end{restatable*}

%% file: theorems/approach1.tex
\begin{restatable*}[\realdc{} with Amortized Locality Approaching $1$]{corollary}{approach} \label{corr:realdcapproach1}
    Let $\Rate{}(\delta) = \delta \times 2^{c_0 \log^3(\log 1/\delta)}$ for some universal constant $c_0 \in \N.$ For any $n = 2^r$, $r \in \N, \eps < 1,$ $\nu > 0,$ and $\delta \in (0,1)$ such that $\Rate{}(\delta) = o(1),$  there exists an (explicit) $(n,k)$ code $\calC$ that is a $\left(\alpha(\delta,\nu), \kappa(\nu) , \delta, \eps\right)$-$\realdc{}$, where $\alpha(\delta,\nu) = (1 + \nu) \times (1 + o(1))$, $\kappa(\nu) = \frac{ \log(1/\eps)\times(\log n)^2\times 2^{(\log\log\log n)^4}}{\nu}$ and rate $1 - \rate{\realdc{}} \geq 1 - o(1)$. 
\end{restatable*}

%% file: sections/preliminaries.tex
In this section, we introduce our notation and any preliminaries from prior work. 
Specifically, we discuss expanders, vicinity locally testable codes, and nested codes.

Let $\N = \{1,2,\dots\}.$ 
Let $\Sigma$ denote the alphabet.
In this work we assume the alphabet is binary, i.e., $\Sigma = \{0,1\}$. 
All words $\bw \in \Sigma^n$ are emphasized in bold lettering. 
The character of word $\bw$ at index $i$ is denoted as $\bw[i].$ 
Generally, for any (ordered) subset of indices $S \subseteq [n], \bw[S] = [\bw[i] : i \in S ].$ 
For words $\bw,\by \in \Sigma^n$, define their relative distance as $\reldist(\bw,\by) := d(\bw,\by) / n.$
Further, for a $(n,k)$ code $\calC,$ define the closest relative distance from word $\bw$ to any codeword in $\calC$ as $\reldist(\bw,\calC) := \min_{\bx \in \Sigma^k} \reldist(\bw,\Enc(\bx)).$
Let $(\bw_1 \circ \bw_2)$ be the usual concatenation of words/vectors.
Define code/vector-space concatenation, where $\calC_1 \circ \calC_2 = \{\by_1 \circ \by_2 \mid \by_1 \in \calC_1, \by_2 \in \calC_2\}.$
Lastly, for a code $\calC,$ we define code exponentiation for any integer $r \geq 1$ to be $\calC^r = \underbrace{\calC \circ \calC \circ \dots \circ \calC}_{r \text{ times}}$.

First, we recall the definition of a linear code.
\begin{definition}[Linear Code]
    A linear $(n,k)$ code (over alphabet $\Sigma$) is a subspace $\calC \subseteq \Sigma^n$ of dimension $k$.
    The linear code $\calC$ has (global) relative distance $\delta$ if for all $\by_1 \neq \by_2 \in \calC,$ $\reldist(\by_1, \by_2) \geq \delta.$
\end{definition}
We say that a $(n,k)$ code has rate $R = (k/n)$.
We call vectors $\by \in \calC$ the {\em codewords of code $\calC$}, or simply {\em codewords} when the code in context is clear.
Furthermore, since all the codes in this paper are linear, we will any linear code $\calC$ as just code $\calC$, for brevity.

A $(n,k)$ code $\calC \subseteq \Sigma^n$ is equivalently defined by a (linear and injective) encoding map $\Enc{} : \Sigma^k \rightarrow \Sigma^n$. 
Then, for any codeword $\by = \Enc(\bx)$, where $\bx \in \Sigma^k$, we say that $\bx$ is the message of codeword $\by.$

In this work, we focus on a variant of relaxed locally decodable codes (\rldc{}s), which are codes defined by an additional (relaxed) decoding procedure $\Dec.$
For \rldc{}s, the decoding procedure $\Dec^{\bw}(i)$, on input index $i \in [k]$, make at most $\ell$ queries to a (possibly corrupted) codeword $\bw$ with at most a $\delta$ fraction of errors, and outputs either (1) the message symbol $\bx[i]$ or (2) $\bot$ with high probability $1 - \eps.$
We call such a code a $(\ell,\delta,\eps)$-\rldc{}. 
Another natural variant of \rldc{}s are known as relaxed locally {\em correctable} codes (\msf{rLCC}s). 
Instead of decoding message symbols, \msf{rLCC}s decode codeword symbols.

\begin{remark}
    In certain cases, it will be convenient to view code $\calC$ as its subset definition $\calC \subseteq \Sigma^n$ (e.g., reasoning about nested codes) and in other cases, it will be more convenient to view code $\calC$ as its algorithmic definition as a tuple of encoding and decoding algorithms $(\Enc, \Dec)$.
In this work, we overload the definition of code $\calC$ as its subset definition $\calC \subseteq \Sigma^n$ and its algorithmic definition as a tuple of encoding and decoding algorithms $\calC = (\Enc, \Dec)$, or, in the case of locally testable codes, $\calC = (\Enc,\Test)$. 
\end{remark}

We introduce a variant on relaxed locally decodable codes called {\em amortized} relaxed locally decodable codes (\realdc{}s).
Instead of the decoder $\Dec$ receiving input index $i \in [k]$ and outputting a single message symbol $\bx[i]$ (or outputs $\bot$), the relaxed decoder instead receives input {consecutive subset} $[L,R]$ and outputs the corresponding message symbols $\bx[L,R]$ (or outputs $\bot).$
\begin{definition}[\realdc{}]\label{def:realdc}
    A $(n,k)$ code $\calC = (\Enc,\Dec)$ (over $\Sigma$) is a $(\alpha, \kappa, \delta,\eps)$-amortized relaxed locally decodable code (\reaLDC) if 
    \begin{enumerate}
        \item (Completeness) For all $\bx \in \Sigma^k$, and all intervals $[L,R] \subseteq [k]$ of size $R-L+1 \geq \kappa$,
        \[\Pr\lbrack\msf \Dec^{\Enc(\bx)}(L,R) = \bx[L,R] \rbrack = 1,\]
        \item (Soundness) For all $\bx \in \Sigma^k$, all intervals $[L,R] \subseteq [k]$ of size $R-L+1 \geq \kappa$, and all words $\bw \in \Sigma^n$ such that $\reldist(\bw,\Enc(\bx)) \leq \delta$,
        \[\Pr\lbrack\msf \Dec^{\bw}(L,R) \in \left\{\bx[L,R], \bot\right\} \rbrack \geq 1 - \eps,\]
        \item (Locality) For all $\bw \in \Sigma^n$ and all intervals $[L,R] \subseteq [k]$ of size $R-L+1 \geq \kappa$, $\Dec^{\bw}(L,R)$ makes at most $\alpha \times (R - L + 1)$ queries to $\bw.$
    \end{enumerate}
\end{definition}
We remark that the most general amortized locally decodable codes (\aLDC{}) definition by Blocki and Zhang \cite{ISIT:blozha25} considered the decoding of a more general subset $Q \subseteq [k]$ of size $|Q| \geq \kappa$ rather than the decoding of a consecutive interval
$[L,R].$
However, all prior ideal \aLDC{} constructions (for Hamming and ins-del errors) require the decoding of intervals instead of arbitrary subsets \cite{ITC:bloZha25,ISIT:blozha25}, indicating that consecutive interval decoding may be necessary for ideal parameters (constant amortized locality, constant rate, and constant error tolerance).
Similar to the non-relaxed \aldc{} setting, we leave the construction of non-consecutive decoding \realdc{}s as an open question for future research.

We additionally remark that non-trivial relaxed decoding is reliant on the completeness property.
In this work, we consider {\em perfect} completeness, i.e., any (non-corrupted) codeword $\by = \Enc(\bx)$ is always locally decoded. 

 As a stepping stone to \realdc{}s, we construct a variant of \realdc{}s that recovers codeword symbols instead of message symbols. 
Following naming conventions of prior work, we call these codes relaxed amortized locally {\em correctable} codes (\msf{aRLCC}s). 

\begin{definition}[\realcc{}]\label{def:realcc}
    A $(n,k)$ code $\calC = (\Enc,\Dec)$ (over $\Sigma$) is a $(\alpha, \kappa, \delta,\eps)$-amortized relaxed locally correctable code (\realcc{}) if
    \begin{enumerate}
        \item (Completeness) For all $\bx \in \Sigma^k$, and all intervals $[L,R] \subseteq [n]$ of size $R-L+1 \geq \kappa$,
        \[\Pr\lbrack\msf \Dec^{\Enc(\bx)}(L,R) = \by[L,R] \rbrack = 1,\]
        \item (Soundness) For all $\bx \in \Sigma^k$, all intervals $[L,R] \subseteq [n]$ of size $R-L+1 \geq \kappa$, and all words $\bw \in \Sigma^n$ such that $\reldist(\bw,\Enc(\bx)) \leq \delta$,
        \[\Pr\lbrack\msf \Dec^{\bw}(L,R) \in \left\{(\Enc(\bx))[L,R], \bot\right\} \rbrack \geq 1 - \eps,\]
        \item (Locality) For all $\bw \in \Sigma^n$ and all intervals $[L,R] \subseteq [n]$ of size $R-L+1 \geq \kappa$, $\Dec^{\bw}(L,R)$ makes at most $\alpha \times (R - L + 1)$ queries to $\bw.$
    \end{enumerate}
\end{definition}

\

\subsection{Expander Codes and \vltc{}s}
We reintroduce prior results for expander and expander graphs. 
Let $G = (L,R,E)$ be an undirected bipartite graph. 
Let $\deg(v)$ be the degree of vertex $v \in L \cup R$ in $G$ and similarly define $\deg(S)$ for any subset $S \subseteq L \cup R.$

For any $v \in L \cup R,$ define $N(v) := \{u \mid (u,v) \in E\}$ to be the set of neighbors of $v$ in graph $G.$ 
Naturally, for any subset $S \subseteq L \cup R,$ define the set of neighbors of $S$ as $N(S) := \cup_{v \in S} N(v)$.
For any subset of left vertices $S \subset L$, define the {\em unique neighbor set} (with respect to the left vertices) as \[ N_u(S) := \{ y \in R \mid |N(y) \cap S| = 1\},\]
i.e., $N_u(S) \subset R$ are the right-neighbors of left-vertices $S \subseteq L$ that have exactly one edge connecting to a vertex in $S$.  

An expander graph is a bipartite graph with an expansion property: any subset of left-vertices $S \subseteq L$ with size at most $\delta |L|$ for some $\delta \in (0,1)$ has a large expanding neighbor set with size $\zeta N(S)$ for some $\zeta \in (0,1).$
Similarly, one can define a unique expander graph by replacing the neighbor set $N(S)$ with the  unique neighbor sets $N_u(S)$ within the prior expansion property.
\begin{definition}[Unique Expander Graph]
    A bipartite graph $G = (L,R,E)$ is a $(d,\delta,\zeta)$ unique expander if (1) it is $d$-left regular (for all $v \in L, \deg(v) = d$), and (2) for every $S \subset L$ of size $|S| \leq \delta |L|,$ it holds that $|N_u(S)| \geq \zeta d|S|.$
\end{definition}
Explicit unique expander graphs follow from the lossless expander construction by Capalbo, Reingold, Vadhan, and Wigderson in \cite{STOC:CRVW02}.
\begin{restatable}[\protect{\cite[Theorem~7.3]{STOC:CRVW02}}]{theorem}{explicitExpander}\label{thm:explicitExpander}
    There exists universal constants $c_0 \geq 1$ and $\beta \leq 1$ such that for all $n \in \N$ and $m \leq n,$ there exists an explicit $(d,\delta,\zeta)$ unique expander $G = (L,R,E)$, where 
    \begin{itemize}
        \item $|L| = 2^n$, $|R| = 2^m$
        \item $d \leq 2^{c_0 \log^3(n - m)}$, $\delta = \beta \times \left( 2^{m - n} / d\right),$ and $\zeta = \Omega(1).$
    \end{itemize}
\end{restatable}
Expander graphs can be used to define codes by viewing the right-vertices as parity-checks/constraints on the codeword symbols labeled by the left-vertices.
We formally define expander graph-based codes, or {\em expander codes}.
\begin{definition}[Expander Code]\label{def:expander}
    Let $G = (L,R,E)$ be a bipartite graph. 
    An expander code $\calC$ associated with bipartite graph $G$ is defined as
    \[\calC = \calC(G) := \left\{ \by \in \Sigma^{|L|} \  \middle | \ \forall v \in R. \sum_{u \in N(v) } \by[u] = 0 \right\}.\]
\end{definition}
Alternatively, we can define the encoding scheme for an expander code by encoding each parity check as a row in a matrix $H$ and define the generator matrix as the inverse of $H$. 
We formally define a procedure called $\msf{GraphToCode}$, which takes an expander graph $G$ as input and outputs the encoding procedure $\Enc{}.$
Note that this is defined separately to $\Enc{}$ to emphasize that the encoding scheme can be precomputed prior to any message input.
\begin{betterFrame}{$\msf{GraphToCode}(G = (L,R,E))$}
\textbf{Input:} Bipartite graph $G = (L,R,E)$.

\textbf{Output:} Encoding procedure $\Enc_G : \Sigma^k \rightarrow \Sigma^n$.
\begin{enumerate}
    \item Interpret vertices $L = \{u_1,\dots,u_n\}$ and $R = \{v_1,\dots,v_m\}$.
    \item For each $v_j \in R:$
    \begin{enumerate}
        \item Compute row vector $\vec{h}_j \in \Sigma^n$ such that for each entry $i \in [n],$ $\vec h_{j,i} = \begin{cases}
            1 & \text{if $u_i \in N(v_j)$;} \\
            0 & \text{otherwise.}
        \end{cases}$
        \item Compute matrix $H = \m{\mathrel{\vcenter{\hbox{\rule{.5cm}{.75pt}}}} \ \vec h_1 \ \mathrel{\vcenter{\hbox{\rule{.5cm}{.75pt}}}}\\ \mathrel{\vcenter{\hbox{\rule{.5cm}{.75pt}}}} \ \vec h_2 \ \mathrel{\vcenter{\hbox{\rule{.5cm}{.75pt}}}} \\   \vdots \\ \mathrel{\vcenter{\hbox{\rule{.5cm}{.75pt}}}} \ \vec h_m \ \mathrel{\vcenter{\hbox{\rule{.5cm}{.75pt}}}}} \in \Sigma^{m \times n}$.
        \item Compute generator matrix $C \in \Sigma^{k\times n}$ such that $H C^T = 0$.
        \item Output procedure $\Enc_G : \Sigma^k \rightarrow \Sigma^n$ defined as $\Enc_G(\bx) := \bx C.$
    \end{enumerate}
\end{enumerate}
    
\end{betterFrame}

\paragraph{Expander Codes are Vicinity Locally Testable} A related type of code is known as a locally testable code.
These are codes where there is a local testing algorithm $\Test{}$, which given any corrupted word $\bw,$ is able to detect error with probability proportional to the number of errors.
A crucial observation in prior work \cite{CCC:CohYan24} was that although expander codes are not fully locally testable, they are locally testable within the vicinity of their distance. 
That is, expander codes have a local testing algorithm for words $\bw$ whose distance from any codeword $\by$ (i.e., number of corruptions) is at most a predefined upper bound $\delta$ (traditional \LTC{}s do not have this distance constraint).
\begin{definition}[\vltc{} \cite{CCC:CohYan24}]\label{def:vltc}
    A $(n,k)$ code $\calC = (\Enc,\Test)$ is a $(\ell,\delta,\rho,\sigma)$-\vltc{} if:
    \begin{itemize}
        \item (Completeness) For any $\bx \in \Sigma^k$, $\Test^{\Enc(\bx)} = \Accept$ with probability $1$;
        \item (Soundness) For any $\bx \in \Sigma^k$, $\by = \Enc(\bx)$, and $\bw \in \Sigma^n$ such that $\reldist(\by,\bw) \leq \delta,$
        \[\Pr\left[\Test^{\bw} = \bot\right] \geq \rho  \times \reldist(\by,\bw) - \sigma;\]
        \item (Locality) Given oracle access to any word $\bw \in \Sigma^n,$ $\Test^{\bw}$ always makes at most $\ell$ queries to $\bw$.
    \end{itemize}
\end{definition}
Unique expander graphs give rise to nontrivial \vltc{}s, i.e., the regime where $\sigma <  \rho \delta$.

\begin{restatable}[Unique Expanders are vicinity-testable \cite{CCC:CohYan24}]{lemma}{vtester}\label{lem:vtester}
    Suppose that $G = (L,R,E)$ is a $(d,\delta,\zeta)$-unique expander with average right degree $\bar c.$
    Let $\Enc = \Enc_G := \msf{GraphToCode}(G)$.
    Then, for every $\eta > 1,$ there exists a randomized procedure $\Test_\eta = \Test_{\eta,G}$ such that
    \begin{itemize}
        \item For every $\bx \in \Sigma^k, \by = \Enc(\bx),$ and $\bw \in \Sigma^n$ such that $\reldist(\by,\bw) \leq \delta,$
        \[\Pr\left[\Test_{\eta}^{\bw} = \bot\right] \geq \zeta \bar c  \times \reldist(\by,\bw) - \frac{1}{\eta};\]
        \item  For every $\bx \in \Sigma^k$, $\Test_{\eta}^{\Enc(\bx)} = \Accept$ with probability $1$;
        \item Given oracle access to any word $\bw \in \Sigma^n,$ $\Test_{\eta}^{\bw}$ always makes at most $(b \bar c)$ queries to $\bw$.
    \end{itemize}
    In particular, for any $\eta > 1,$ code $\calC = (\Enc_G,\Test_{b,G})$ is a $(\ell = \eta\bar c, \delta, \rho = \zeta \bar c,\sigma =  1/ \eta)$-\vltc{}.
\end{restatable}
\begin{proof}
Construct a tester $\Test_{\eta,G}$ as follows:
    \begin{betterFrame}{$\Test_{\eta,G}^{\bw}$}
    
    \textbf{Parameters}: Integer $\eta > 1$ and bipartite $G = (L,R,E).$
    
    Let $R' = \left \{v  \in R \middle | \deg (v) \leq \eta \bar c\right\}.$
        \begin{enumerate}
            \item Sample $v \in R'$ uniformly at random.
            \item Query $\bw$ at the indices $N(v) \subset L$, i.e., for all $u \in N(v),$ query $\bw[u]$.
            \item Output $\Accept$ if $\sum_{u \in N(v)}\bw[u] = 0$ and $\bot$ otherwise.
        \end{enumerate}
    \end{betterFrame}
     First, $|R'| \geq (1 - 1/\eta) |R|$ since otherwise, there are more than $(1/\eta) |R|$ right-vertices that have degree greater than  $\eta\bar c.$
    We sample a single right-vertex $v \in R',$ so the number of queries made in step~$2$ is $\eta \bar c.$
    Further, when $\bw = \Enc(\bx)$ for some $\bx \in \Sigma^k$ the tester \Test{} always succeeds by construction of $\Enc$.

    Let $\bw \in \Sigma^{|L|}$ be a word such that $\reldist(\by,\tilde{\by}) \leq \delta$ for some codeword $\by.$ 
    Let $S = \{v \in L \mid \bw[v] \neq \by[v]\}.$
    By the unique expansion property, since $|S| \leq \delta |L|,$ we have $|N_u(S)| \geq \zeta d |S|.$
    If the vertex $v$ sampled in step $1$ is in $N_u(S),$ then $\Test{}$ outputs $\bot.$ 
    The probability that $v \in N_u(S)$ is at least
    \begin{align*}
        \frac{|N_u(S)| - |R / R'|}{|R|} &\geq \frac{\zeta d |S| - (1/\eta)|R|}{|R|}\\ 
                                        &\geq \zeta \times \frac{d |S|}{|R|} - \frac 1 \eta\\
                                        &\geq \zeta \bar{c} \times \reldist(\bw,\by) - \frac{1}{\eta}. \qedhere
    \end{align*}

\end{proof}
\subsection{Nested Codes} 
We reintroduce the nesting operation introduced by Kumar and Mon \cite{STOC:KumMon24} which was used to recursively construct relaxed locally decodable codes from locally testable codes.
\begin{definition}[Nested Code \cite{STOC:KumMon24}]\label{def:nested}
    Let $\calC \subseteq \Sigma^N$ be an $(N,K)$ code and $\calC' \subseteq \Sigma^n$ be a $(n,k)$ code, where $N = bn$ for some $b \in \N$.
    The nesting operation $\divideontimes : 2^{\Sigma^N } \times 2^{\Sigma^n} \rightarrow 2^{\Sigma^N}$ is defined as
    \[ \calC \divideontimes \calC' := \calC \cap \left((\calC')^{\lfloor N / 2 \rfloor} \times \Sigma^{N - \lfloor N / n \rfloor \times n}\right) \cap \left(\Sigma^{N - n} \times \calC'\right).\]
    We say code $\calC'' = \calC \divideontimes \calC'$ is a nested code (of $\calC$ and $\calC'$).
\end{definition}

\begin{remark}
Our constructions are based off the nested expander code construction of Cohen and Yankovitz \cite{CCC:CohYan24}, where they use expander codes whose code length double, i.e., $(N,K)$ code $\calC$ is nested with $(n,k)$ code $\calC'$ where $N = 2n.$
In general, when $N = bn$ for some $b \in \N,$ the nested code definition simplifies to
\[\calC \divideontimes \calC' = \calC \cap (\underbrace{\calC' \circ \dots \circ \calC'}_{b \text{ times}}).\]
\end{remark}
\begin{lemma}[Nested Codes: Rate and Distance \cite{STOC:KumMon24}]\label{lem:nestingParams}
    Let $\calC \subseteq \Sigma^N$ be a code with rate $1 - \Rate{}$ and distance $\delta$. 
    Let $\calC' \subseteq \Sigma^n$ be a code with rate $1 - \Rate{}'$ and distance $\delta'$ such that $N = bn$ for some $b \in \N$. Then, $\calC \divideontimes \calC'$ is a nested code with rate at least ${1 - R - R'}$ and (global) relative distance $\max\{\delta,\delta' / b\}.$
\end{lemma}
We additionally define recursively-nested codes, which apply the nesting operation recursively on a sequence of codes $\{\calC_j\}_{j \in \N}.$
\begin{definition}[Recursively-Nested Code]\label{def:recNested}
Code $\calC$ is a $s$-recursively-nested code under code sequence $\{\calC_i\}_{i = 0,1,2,\dots}$ if (1) $s = 0$ and $\calC = \calC_0$ or (2) for $s \geq 1$, $\calC = \calC_s \divideontimes \calC'$ where $\calC'$ is a $(s-1)$-recursively-nested code (under the same code sequence).
Then, we denote an $s$-recursively-nested code $\calC$ as \[\calC = \mathlarger{\mathlarger{\divideontimes}}^s_{j=1} \calC_j.\]
\end{definition}
    Intuitively, a nested code $\calC \divideontimes \calC'$ corresponds to a $1$-recursively-nested code (under a code sequence $\{\calC_j\}$ where $\calC_0 = \calC'$ and $\calC_1 = \calC$). 
    Naturally, the rate and distance of (single) nested codes in Lemma~\ref{lem:nestingParams} can be extended to recursively-nested codes. 
\begin{corollary}[Recursively Nested Codes: Rate and Distance  \cite{STOC:KumMon24}]\label{corr:nestedRate}
    Let $\{\calC_j\}_{j \geq 0}$ be a sequence of $(n_i,k_i)$ codes with corresponding rates $\{1 - \Rate{j}\}_{j \geq 0}$, relative distances $\{\delta_j\}_{j \geq 0}$, and codeword lengths $\{n_j\}_{j \geq 0}$ such that for some $b \in \N$, $n_{j + 1} = bn_j$ for all $j \geq 0$.
    Then, for all $s \geq 0$, the $s$-recursively nested code of code sequence $\{\calC_j\}_{j \geq 0}$ has rate at least $1 - \sum_{j = 0}^r \Rate{j}$ and (global) relative distance $\max\limits_{0 \leq j \leq r } (\delta_j / b^{s-j}).$
\end{corollary}

%% file: sections/fbdc.tex
We introduce the notion of fuzzy block decoding codes (\FBDC{}) to modify and analyze the \rlcc{} of Cohen and Yankovitz \cite{CCC:CohYan24} with respect to amortizing locality.
Essentially, we will need a code with relaxed local correctability of predefined codeword blocks of size $B,$ where we allow the `block decoder' to either output (1) the decoded blocks up to some fraction of errors $\delta_B$, or (2) $\bot$ on the detection of error.

\begin{definition}[\FBDC{}]\label{def:FBDC}
    A linear $(n,k)$ code $\calC = (\Enc,\BlockDec)$ (over $\Sigma$) is a $(\ell,B,\delta,\delta_B,\eps)$-\FBDC{} if $B \mid n$, and for all $\by \in \calC,$ where $\by = \by_1 \circ \dots \circ \by_{n / B}$ such that $|\by_j| = B$ for all $j \in [n / B]$,
      \begin{enumerate}

        \item (Completeness) For all $j \in [n / B],$ \[\Pr\lbrack\msf \BlockDec^{\by}(j) = \by_j \rbrack = 1,\] 
        \item (Soundness) For all words $\bw \in \Sigma^n$ such that $\reldist(\bw,\by) \leq \delta$ and all $j \in [n / B]$,
        \[\Pr\lbrack  \BlockDec^{\bw}(j) \in \calW_j \cup \{\bot\} \rbrack \geq 1 - \eps,\]
         where $\calW_j = \left\{\bz  \in \Sigma^{B} \mid \reldist(\bz,\by_j) \leq \delta_B \right\}$,
         \item (Verbatim) For any $\bw \in \Sigma^n$ and all $j \in [n / B],$ 
        \[\Pr\lbrack\msf \BlockDec^{\bw}(j) \in \{\bw_j,\bot\} \rbrack = 1,\] 

        \item (Locality) Given oracle access to any word $\bw \in \Sigma^*,$ $\BlockDec^{\bw}$ makes at most $\ell$ queries to $\bw.$

    \end{enumerate}
\end{definition}
We emphasize the distinction of the top-level error tolerance $\delta$ and block-level error tolerance $\delta_B.$ 
The soundness of the block decoder $\BlockDec{}^{\bw}(j)$ is guaranteed whenever the queried word $\bw$ has at most $\delta$ fraction of errors from some codeword $\by$, while the outputted symbols from the block decoder have at most $\delta_B$ fraction of errors from the true codeword block $\by_j$.
{Note that we require a `verbatim' property, where the block decoder $\BlockDec{}^{\bw}(j)$ either outputs the specific block $\bw_j$ with no modifications or outputs $\bot$. 
The verbatim property implies that once word $\bw$ and index $j$ are fixed, the output of $\BlockDec{}^{\bw}(j)$ is also deterministically equal to $\bw_j$,conditioned on $\BlockDec{}^{\bw}(j) \neq \bot$. 
In particular, whether a non-$\bot$ output of $\BlockDec^{\bw}(j)$ lies within the required block distance $\delta_B$ is determined prior to any decoder randomness.
This will be crucial in avoiding a union bound in the soundness error of the \realcc{} (Theorem~\ref{thm:FBDCtoARLCC}), since the \realcc{} decoder will call the \FBDC{} decoder multiple times.}

In the following proposition, we show how to construct an \FBDC{} from a family of (unique) expander graphs by recursively nesting their respective expander codes. 
The block decoding will follow from recursive application of the \vltc{} testers of the expander codes. 

\begin{restatable}[Expander to \FBDC{}]{proposition}{recursiveExpander} \label{prop:recursiveExpander}
    
    Let $\Rate{} \in (0,1)$. 
Let $\{G_i = (L_i,R_i,E_i)\}_{i \in \N}$ be a sequence of bipartite graphs such that:
\begin{enumerate}
    \item (Parameters) $|R_i| / |L_i| = \Rate{}$ and $|L_{i + 1}| = 2 |L_i|$ for all $i \in \N$,

    \item (Expanders) $\{G_i\}_{i \in \N}$ is a family of $(d,\delta,\zeta)$ unique expander graphs with average right degree $\bar c$.
\end{enumerate}

    Then, for any $s \in \N$, $n = 2^{s} \times |L_1|$, $g \in \N$, and any $\eta > 1,$ there exists a $(n,k)$ code $\calC_{\FBDC{}}$ that is a $(\ell,B,\delta,\delta_B ,\eps)$-\FBDC{} with 
    \begin{itemize}
        \item $\ell = \eta \bar cgs + |L_1| /2$,
        \item $B = |L_1| / 2$, 
        \item $\delta = \delta / 2$, 
        \item $\delta_B = 2\delta$,
        \item $\eps \leq \exp\left(g\left(\frac{1}{\eta} - \frac{\zeta \delta}{2} \right)\right)$, and
        \item rate $1 - \Rate{\FBDC{}} \geq 1 - s\Rate{}$.
    \end{itemize}
\end{restatable}
\begin{proof}
We construct the encoding procedure $\Enc_{\FBDC{}}$ as follows:
\begin{betterFrame}{$\Enc_{\FBDC{}}$ Construction}
    Let $\calP_1,\dots,\calP_s$ be a fixed sequence of partitions of $[n]$ such that for every $i = 1,\dots,s$, $\calP_i$ has $2^{i}$ equal-size parts denoted $\{\calB_1^i,\dots,\calB^i_{2^{i}}\}$, where $\calB_j^i = \{(j - 1) \times (n/2^{i}) + 1,\dots, j \times (n / 2^i)\}.$

Form expander graph $G = (L,R,E)$ as follows:
\begin{enumerate}
    \item Set the vertices,
    \[L \leftarrow [n] \text{ and } R \leftarrow \bigcup_{i \in [m]} R_i.\]
    (note that we choose the labeling of each right vertex to be unique, i.e., for each ${i \neq j \in [s], R_i \cap R_j = \emptyset}$).
    \item For each $i \in [s], \calB_j \in \calP_i,$ define function
    \[\msf{proj}_{i,j}(u) := (j - 1) \times (n / 2^i) + u,\] 
    and compute edges
    \[E_{i,j} = \left\{\left(\msf{proj}_{i,j}(u),v\right) \mid (u,v) \in E_{s - i + 1}\right\},\]
    Set edges, 
    \[E \leftarrow E  \cup E_{i,j}.\]
    \item Output $\Enc_{\FBDC{}} \leftarrow \msf{GraphToCode}(G = (L,R,E))$.
\end{enumerate}
\end{betterFrame}

Observe that the encoding procedure $\Enc_{\FBDC{}}$ described above directly corresponds to the $s$-recursively nested code $\calC_{\FBDC{}} = \divideontimes^{s}_{i = 0} \calC_i,$ where $\calC_i = \msf{GraphToCode}(G_i)$ for all $i \geq 1$ and $\calC_0 = \Sigma^n.$
To see this, let $\calC^{\perp}$ denote the dual code of $\calC$.
It is well known that (1) for any linear codes $\calC_1,\calC_2,\dots,\calC_s$ that \[\left(\bigcap_{i \in [s]} \calC_i \right)^\perp = \text{span}\left\{\bigcup_{i \in [s]}\calC_i^{\perp}  \right\},\]
and (2) if $\calC$ is linear, then $\calC^{\perp}$ is also linear with generator matrix corresponding to the parity check matrix of $\calC.$
Then, for our nested code, observe that $\divideontimes_{i = 1}^{s}\calC_i$ has dual code \[\left(\divideontimes_{i = 1}^{s}\calC_i\right)^{\perp} = \text{span}\left\{\bigcup_{i = 0}^{s}(\calC_i^{2^i})^\perp \right\},\]
i.e., the dual code is formed by taking constraints given by $\calC_i^\perp$ on every codeword block specified by partition $\calP_i$ for every $i \in [s].$
These are exactly the constraints added to the expander graph in the $\Enc_{\FBDC}$ construction.
Hence, $\Enc_{\FBDC}$ is an encoding procedure of code $\calC_{\FBDC}$ and by Corollary~\ref{corr:nestedRate}, code $\calC_{\FBDC}$ has rate $1- \Rate{} \geq 1 - s\Rate{}$.

We now construct the decoding procedure $\BlockDec$.
For each $i \in [s],$ let $\Test_i$ be the testing procedure of code $\calC_i$ that is a $(\eta \bar c,\delta,\zeta \bar c, 1 / \eta )$-\vltc{} (Lemma~\ref{lem:vtester}). 
Then, $\BlockDec$ is defined as:
\begin{betterFrame}{$\BlockDec^{\bw}(j)$}
    Let $j_1,\dots,j_{s}, j_{s+1}$ and  be the sequence of indices such that $\calB_{j_1}^1 \supset \calB_{j_2}^2 \supset \dots \supset \calB_{j_s}^s \supset \calB_{j_{s+1}}^{s+1} = \calB_{j}^{s+1}$.
        \begin{enumerate}
        \item For $i = 1,\dots,s$:
        \begin{enumerate}
            \item Set $\bw_i \leftarrow \bw\left[\calB_{j_i}^i\right].$
            \item Run $\Test_i^{\bw_i}$ $g$ times. 
            If any tests outputs $\bot,$ output $\bot.$
        \end{enumerate}
        If $\bot$ was not outputted on any iteration, output $\bw\left[\calB_{j}^{s+1}\right]$.
    \end{enumerate}
\end{betterFrame}
    We argue correctness.
    Fix $j \in [n / B].$
    For any codeword $\by \in \calC$ and $i \in [s + 1]$, define $\by_i = \by\left[\calB_{j_i}^i\right]$.
    First, observe that if the decoder $\BlockDec^{\by}$ is given access to a codeword $\by \in \calC,$ it will always output the correct codeword symbols $\bw_{s+1} = \by_{s+1}$ with probability $1$ by the completeness of the \vltc{} testers $\Test_i.$
    Hence, completeness holds.
    
    Now, for soundness, suppose that $\BlockDec^{\bw}$ is instead given access to a word $\bw$ such that $\reldist(\bw,\by) \leq \delta / 2$. 
    In the case that $\reldist(\bw_{s+1}, \by_{s+1}) \leq 2\delta$, regardless of whether a tester $\Test_i$ outputs $\bot$ on any iteration $i \in [s]$ or not, $\BlockDec^{\bw}$ satisfies \FBDC{} soundness trivially.
    Hence, we suppose that $\reldist(\bw_{s+1}, \by_{s+1}) > 2\delta$ and derive the probability that the decoder outputs $\bot.$
    
    Since $\reldist(\bw,\by) \leq \delta / 2$ and $|\bw_1| = |\bw| / 2,$ this implies $\reldist(\bw_1,\by_1) \leq \delta$.
    Next, since $\reldist(\bw_{s+1}, \by_{s+1}) > 2\delta$, this implies $\reldist(\bw_{s},\by_{s}) > \delta.$
    There exists a level $i^* \in [s]$  such that, 
    \[\reldist(\bw_{i^*}, \by_{i^*}) \leq \delta \text{ but } \reldist(\bw_{({i^*}+1)}, \by_{({i^*}+1)}) \geq \delta.\]
    Moreover, since $\left|\calB^{i}_{j_i}\right| = 2\left|\calB^{i+1}_{j_i+ 1}\right|$ for all $i \in [s],$ we have that
    \[\delta / 2 \leq \reldist(\bw_{i^*}, \by_{i^*}) \leq  \delta.\]
    By Lemma~$\ref{lem:vtester}$, a cannonical amplification argument, and a union bound, the soundness error, i.e., when none of the $g$ invocations of $\Test_i$ output $\bot$, is at most $1 -(\frac{\zeta \delta}{2} - \frac{1}{\eta})^g \leq \exp\left(-g\left(\frac{\zeta \delta}{2} - \frac{1}{\eta}\right)\right).$

    Lastly, the queries made are (1) from the \vltc{} testers $\Test_{i}$ and (2) to recover the final block $\bw\left[\calB_j^{s+1}\right].$
    For (1), there are at most $s$ invocations of \vltc{} testers $\Test_i$, each with locality $\eta \bar c$. 
    For (2), each block has size $|L_1|$.
    Thus, the total query complexity is $\eta \bar c s   + |L_1|.$
    \qedhere
\end{proof}
In our desired construction, we will recursively nest expander codes for $O(\log n)$ steps. 
Consequently, the relative distance of the expander codes must be sub-constant, and the overall top-level error tolerance is sub-constant. 
The following Lemma shows how we can raise the error tolerance to a constant fraction via applying a top-level code with sufficiently high error tolerance.
\begin{restatable}[Top-Level Amplification]{lemma}{topLevelAmp}
\label{lem:topLevelAmp}
    Let $(n,k)$ code $\calC' = (\Enc,\BlockDec)$ be a $(\ell',B,\delta',\delta_{B},\eps')$-\FBDC{} and $(2n,k_0)$ code $\calC_{top} = (\Enc_{top},\Test_{top})$ be a $(\ell_{top},\delta,\rho,\sigma)$-\vltc{} with $\delta' < \delta$ and $\rho \delta' / 2 > \sigma$. 
    Then, for any integer $g_{top} \geq 1,$ $\calC_{\msf{amp}} = \calC_{top} \divideontimes \calC$ is a $(g_{top} \ell_{top} + \ell', B, \delta, \delta_B,\eps)$-\FBDC{} where 
    \[\eps \leq \max\left\{\eps', \exp\left(g_{top} \times \left(\sigma - \rho \delta'/2 \right)\right)\right\}.\]
    Moreover, $\calC_{\msf{amp}}$ has rate $1 - \Rate{amp} \geq (k/n) - (1 - k_0/2n).$
\end{restatable}
\begin{proof}
    Define the block decoding procedure $\BlockDec_{\msf{amp}}$  as follows:
    \begin{betterFrame}{$\BlockDec^{\bw}_{\msf{amp}}(j)$}
        \begin{enumerate}
            \item Run $\Test_{top}^{\bw}$ $g_{top}$ times. 
            If any test outputs $\bot$, output $\bot$. Otherwise, 
            \begin{enumerate}
                \item set $(\bw^1,j') \leftarrow \begin{cases}
                    (\bw[1,\dots,n],j) & \text{if }j \leq n / B;\\
                    (\bw[n+1,2n],j - n/B) & \text{if }n / B < j \leq 2n / B, \text{ and}
                \end{cases}$
                \item Output $\BlockDec^{\bw^1}(j')$.
            \end{enumerate}
        \end{enumerate}
    \end{betterFrame}
    Completeness holds by the completeness of $\calC'$ and $\calC_{top}.$
    Now, for soundness, consider word $\bw \in \Sigma^{2n}$ such that $\reldist(\bw,\by) \leq \delta.$
    If $\reldist(\bw,\by) > \delta' / 2,$ since $\calC_{top}$ is a $\vltc{}$, 
    \[\Pr[\Test{}_{top}^{\bw} = \bot] \geq p \reldist(\bw,\by)/ 2 - \sigma,\]
    and so over $g_{top}$ iterations, the probability that all tests accept is at most 
    \[\left(1 - \left(\rho \delta' / 2 - \sigma\right)\right)^{g_{top}} \leq \exp\left(-g_{top}(\rho \delta' / 2 - \sigma)\right).\]
    Otherwise, if $\reldist(\bw,\by) \leq \delta' / 2,$ then $\reldist(\bw_{p},\by_p) \leq \delta')$ for $j \in \{1,2\}$. 
    Then, by soundness of the \FBDC{}, 
    \[\Pr[\BlockDec{}^{\bw_p}(j) \not\in \calW_j \cup \{\bot\}] \leq \eps'.\]
    Combining the two cases, our desired soundness readily follows.

     

    Lastly, we argue the rate, number of queries, and verbatim property.
    By Lemma~\ref{lem:nestingParams}, the rate follows.
    Since $\BlockDec_{\msf{amp}}$ calls $\Test_{top}$ $g_{top}$ times and calls $\BlockDec$ once, the number of queries is $g\ell_{top} + \ell'.$
    Finally, $\BlockDec{}_{\msf{amp}}^{\bw}(j)$ either outputs $\bot$ during one of the top level tests, or outputs $\BlockDec^{\bw_p}(j)$, which by the verbatim property of $\calC'$, is either the block without modifications or $\bot.$
\end{proof}
We show in the following theorem how to transform an \FBDC{} into a \realcc{} by applying an additional nesting step with an appropriately chosen error correcting code.
\begin{restatable}[\FBDC{} to \realcc{}]{theorem}{FBDCtoARLCC}\label{thm:FBDCtoARLCC}
    Let $(n,k_F)$ (linear) code $\calC_{\FBDC{}} = (\Enc_{\FBDC{}},\BlockDec)$ be a $(\ell,B,\delta,\delta_B,\eps)$-\FBDC{} and $\calC_{\msf{Block}} = (\Enc_{\msf{Block}},\Dec_{\msf{Block}})$ be a $(B,k_B)$ (linear) code with error tolerance $\delta_B$. 
    Then, nested code $\calC_{\realcc{}} = \calC_{\FBDC{}} \divideontimes \calC_{\msf{Block}}$ is a $(n,k)$ code that is an $(\alpha,\kappa,\delta,\eps)$-$\realcc{}$ with $\alpha = \frac{3\ell}{B}, \kappa = B$, and rate at least $1 - \rate{\realcc{}} \geq k_F/n -  (1- k_B/B).$

    More generally, for any $h \in [n / B]$, nested code $\calC_{\realcc{}} = \calC_{\FBDC{}} \divideontimes \calC_{\msf{Block}}$ is a $(n,k)$ code that is an $(\alpha_h,\kappa_h,\delta,\eps)$-$\realcc{}$ with $\alpha_h = \frac{(h+2)\ell}{hB}, \kappa_h = hB$, and the same rate $1 - \rate{\realcc{}}$.
\end{restatable}
\begin{proof}
    Construct the amortized decoder as follows:
    \begin{betterFrame}{$\Dec^{\bw}(L,R)$}
    \begin{enumerate}
        \item  Interpret $\bw = \bw_1 \circ \bw_2 \circ \dots \bw_{n / b}$ where each $|\bw_j| = b.$
        \item Let the symbols of $\bw[L,R]$ lie in blocks $\bw_{s+1} \circ \dots \circ \bw_{s + t}.$
        \item Compute $[\Enc_{\msf{Block}}(\Dec_{\msf{Block}}( \BlockDec^{\bw}(s + j)))]_{j \in [t]}$ and output the corresponding bits according to range $[L,R]$.
        If $\bot$ is outputted on any run of $\BlockDec{}^{\bw}$, output $\bot.$
    \end{enumerate}
    \end{betterFrame}

    For correctness, observe that since $\calC = \calC_{\FBDC{}} \divideontimes \calC_{\msf{Block}} \subseteq \calC_{\FBDC{}}$, $\BlockDec{}^{\bw}$ retains completeness from the original \FBDC{} $\calC.$
    That is, computing $\BlockDec{}^{\by}(s + j)$ for each $j \in [t]$ on codeword $\by$ from $\calC$ outputs codeword blocks $\by_{s+1}\circ \dots \circ \by_{s+t}$ with probability $1.$
    Then, completeness of the $\realcc$ follows since $\Enc_{\msf{Block}}(\Dec_{\msf{Block}}(\by_{s + j})) = \by_{s+j}$ for all $j \in [t].$
    
    Now suppose the decoder $\Dec^{\bw}$ is given oracle access to word $\bw$ such that there exists codeword $\by \in \calC$ such that $\reldist(\bw,\by) \leq \delta.$
    Fix any input interval $[L,R]$ and suppose $\Dec^{\bw}(L,R)$ calls $\BlockDec^{\bw}(s + j)$ for $j \in [t]$. 
    Then, since $\Dec$ outputs $\bot$ if any calls to $\BlockDec^{\bw}$ output $\bot,$ the \realcc{} has soundness error at most the soundness error of the \FBDC{}, $\eps.$
    More formally, define the set $\texttt{Bad} = \{j \in [t] : \reldist(\bw_{s + j},\by_{s + j}) > \delta_B\}.$
    Then, if $|\texttt{Bad}| \geq 1,$ we fix an arbitrary $j_{\texttt{Bad}} \in \texttt{Bad}$ and so
    \begin{align*}
        \Pr[\Dec^{\bw}(L,R) \not\in \{\by[L,R], \bot\}] &\leq \Pr[\BlockDec{}^{\bw}(j_{\texttt{Bad}}) \neq \bot] \\
        &= \Pr\left[\BlockDec{}^{\bw}(j_{\texttt{Bad}}) \not\in \left\{\bot,\calW_{s + j_{\texttt{Bad}}}\right\}\right]  \tag{Verbatim property} \\
        &\leq        \eps.
    \end{align*}

    We now argue the locality.
    Let $L\leq R \in [n]$ be codeword indices such that $(R - L + 1) \geq B$. 
    Then, supposing the codeword symbols $\bw[L,R]$ lie in the (minimal) block range $\bw_{s+1} \circ \dots \circ \bw_{s + t}$,
    we have that $t \leq \frac{R - L + 1}{B} + 2$, i.e., the block range may overlap with at most $2$ extra blocks.
    The decoder $\Dec^{\bw}$ runs $\BlockDec{}$ $t$ times. 
    Thus, the total query complexity is $t \ell$ amortized over the number of recovered symbols $(R - L + 1)$, giving 
    \[\alpha \leq \frac{t \ell}{R- L + 1} \leq \frac{\ell\left(\frac{R - L + 1}{B} + 2\right)}{R - L + 1 } = \frac{\ell((R - L + 1) + 2B)}{B(R - L + 1)} \leq \frac{3\ell}{B}.\]
    
\end{proof}

%% file: sections/alccToAldc.tex
In this section, we take our previously constructed amortized locally {\em correctable} codes and transform them into amortized locally {\em decodable} codes. 
Since our $\realcc{}$ codes are linear, we can convert them into $\realdc{}$s by making the code systematic, where recall that a code being systematic means that its codewords include the message.
More formally, any linear code $\calC = (\Enc,\Dec)$ with generator $G \in \Sigma^{k \times n}$ can be made systematic by a invertible row-operation transformation $A \in \Sigma^{k \times k}$ and column permutation matrix $C \in \Sigma^{n \times n}$ such that
\[AGC = \m{I & \mid & P},\]
for identity matrix $I \in \Sigma^{k \times k}$. 
It follows that matrix $(AG)$ is the generator matrix for a code $\calC'$ that is systematic\footnote{The systematic property usually require that the message is encoded in the first $k$ symbols of the codeword, i.e., $\Enc(\bx) = \bx \circ \by'$ for all $\bx \in \Sigma^k$ and some $\by' \in \Sigma^{n -k}$. We relax this ordering (by omitting matrix $C$) since this would require permuting the queries made by the decoder as well. In other words, the locally testable block structure is preserved up to permutations.}.
Then, the systematic code $\calC'$ encodes messages $\bx$ as the codeword $\by = (\bx AG)$ \footnote{We also require that the relative ordering of message symbols is preserved in the systematic code for the \realdc{} decoder. Intuitively this is always possible because we can apply a permutation to the message prior to a systematic encoding i.e., $mT$ for some permutation matrix $T \in \Sigma^{k\times k}$.}.
Observe by associativity, $\by = (\bx A) G,$ and so codeword $\by \in \calC'$ is equivalently a codeword over the original code $\calC$, except with message space $\{\bx A : \bx \in \Sigma^k\}.$
Since $A$ is an invertible, bijective transformation, this new message space is simply a permutation of the original space $\Sigma^k.$
Intuitively, since the systematic code $\calC'$ is equivalent to original code $\calC$ up to message space permutation, the $\realcc{}$ decoder $\Dec$ for original code $\calC$ is also a $\realdc{}$ decoder for the systematic code $\calC'.$

However, the amortized locality parameters $\alpha$ and $\kappa$ need not be preserved because the block structure of our encoding forces the message symbols to be ``distributed'' in the code. 
For instance, recall that our encoding is a recursively nested code, where the first $(n / 2)$ symbols, first $(n /4)$ symbols, and so forth, form a \vltc{}.
If the first $(n/2)$ symbols, or first $(n/4)$ symbols, or so forth, are all message symbols, then these symbols could not possibly form a \vltc{} codeword, since there are no symbols that allow you to distinguish one message from another.
Hence, the message symbols are distributed within each codeword, where the distribution may be unbalanced in certain blocks depending on the exact (expander) code used.
If many consecutive blocks contained a small proportion of message symbols, then the resulting amortized locality of the \realdc{} must also increase inverse-proportionally.

We remedy the situation of disproportionate message symbols per block by describing an additional `filtering' transformation to the \realdc{}.
Intuitively, the lowest level blocks of a recursively nested code must have many blocks with message symbol proportion close to the rate of the code.
Then, by ignoring these `bad' blocks (e.g., fixing the message symbols mapping to these blocks), we can guarantee that if a block contains message symbols, the proportion of message symbols must be high. 
However, this transformation comes at the cost of lowering the code's rate; 
fortunately, we will show that the rate loss is tolerable.

To start, we formally define good (resp. bad) codeword blocks.
Let $B \in \N$ be any block size. 
Then, for any $(n,k)$ code $\calC$ and codeword $\by \in \calC$, denote the i'th block of codeword $\by$ of size $B$ as $\by[(i - 1) B + 1: iB]$. 

Given this context, we simply say (lowest level) block $i$ for brevity.
\begin{definition}[$(c,B)$-good/bad block]\label{def:goodBlock}
    For $c \in [0,1]$ and any $j \in [n / B]$, block $j$ is a {\em $(c,B)$-good block} if it either contains {no message symbols} or at least $cB$ message symbols. 
    Otherwise, we say the block $j$ is a {\em $(c,B)$-bad block}.
\end{definition}
Our goal is to ensure that every block in our codewords are $(c,B)$-good, which we define a code achieving this as a $(c,B)$-message smooth code.
\begin{definition}[$(c,B)$-message smooth code]\label{def:smoothCode}
    A (systematic) $(n,k)$ code $\calC$ is $c$-message smooth if block $j$ is $(c,B)$-good for all $j \in [n/ B]$.
\end{definition}
In the context of our nested \realcc{} construction, the block size $B$ is the lowest level block size. 
For notational clarity, we simply say $c$-good/bad blocks and $c$-smooth codes, when referring to our nested codes with lowest level block size $B.$
We derive the maximum number of bad blocks (equivalently the minimum number of good blocks) a recursively nested, systematic code may have.
Note that we derive a bound looser than necessary since it will be both significant clearer and is sufficient for our constructions.
Specifically, we will assume any block containing {\em exactly} a $c$ fraction of message symbols to be $c$-bad.

\begin{lemma}\label{lem:numGoodBlocks}
    Let $\calC$ be a recursively nested, systematic $(n,k)$ code with lowest level block size $B$ and rate $1 - \Rate{}.$
    Then, for any $c \leq 1 - \Rate{}$, $\calC$ has at least $\frac{k - cn}{(1- c) \times B}$ $c$-good blocks.
\end{lemma}
\begin{proof}
    We proceed with a pigeon-hole argument, maximizing the number of bad blocks.
    Given $k$ message symbols, distribute $c B$ message symbols to each block $j \in [n / B]$. 
    The number of remaining message symbols left to distribute is
    $k - cB\times (n / B) = k - cn,$
    and each block can take at most $B - cB = (1 - c) \times B$ more message symbols.
    Then, the minimum number of good blocks is at least $\frac{k - cn}{(1 - c) \times B}.$
\end{proof}

We show a modification of the original nested code $\calC$ into a message-smooth nested code $\calC_{\msf{smooth}}$ by removing lowest level blocks with small proportion of message symbols.
\begin{theorem} \label{thm:SmoothifyingNestedCodes}
    Let  $\calC_{\msf{nest}}$ be a recursively nested, systematic $(n,k)$ code with lowest level block size $B \in \N$ and rate $1 - \Rate{} \in (0,1]$.
    Then for any constant $c \leq 1 - \Rate{},$ there exists a recursively nested code $\calC_{\msf{smooth}}$ with rate $\frac{1 - \Rate{}  - c}{1 - c}$ and the same lowest level block size $B$ that is $c$-message smooth.
\end{theorem}
\begin{proof}
Let $\Enc$ be the encoding procedure of code $\calC$. 
Then, we construct $\Enc_{\msf{smooth}}$, the encoding procedure for $\calC_{\msf{smooth}}$ as follows:
\begin{betterFrame}{$\Enc_{\msf{smooth}}(\bx)$}
\begin{enumerate}
    \item Compute $\by \leftarrow \Enc\left(0^{k}\right)$.
    \item Let $j_1,\dots,j_p \in [n / B]$ be the indices of all $c$-bad blocks in $\by$. 
    \item Let $(i_1^s,i_1^e),\dots, (i^s_p,i^e_p) \in [k]\times [k]$ be the corresponding (start and end) indices such that $\bx[i_q^s,i_q^e]$ are the message symbols corresponding to the $j_q$ bad block of $\by$ for some $q \in [p]$.
    \item Initialize $\bx'$ as an empty word. For $i \in [k]$:
    \begin{enumerate}
        \item If $i \in [\{i_q^s, i_q^s + 1,\dots, i_q^e\}]$ for some $q \in [p]$: set $\bx' \leftarrow \bx' \circ [0]$;
        \item Otherwise, set $\bx' \leftarrow \bx' \circ \bx[i]$.
    \end{enumerate}
    \item Output $\Enc(\bx').$
\end{enumerate}
\end{betterFrame}

By construction, code $\calC_{\msf{smooth}}$ is $c$-smooth since all $c$-bad blocks in the original code $\calC$ do not contain any message symbols in the new code $\calC_{\msf{smooth}}$.
Next, by Lemma~\ref{lem:numGoodBlocks}, there are at least $\frac{k - cn}{1-c}$ message symbols remaining, so
the rate of the code $\calC_{\msf{smooth}}$ with encoding procedure $\Enc_{\msf{smooth}}$ is at least $\frac{k - cn}{(1 - c) \times n} = \frac{1 - \Rate{} - c}{1 - c}$.
\end{proof}

Naturally, we may extend Theorem~\ref{thm:SmoothifyingNestedCodes} from recursively nested codes to our $\realcc{}$s. 
Intuitively, after transforming the $\realcc{}$ to a $c$-smooth code, each block is guaranteed to have more than a $c$ fraction of message symbols. 
Subsequently, since the decoder goes from recovering entire codeword blocks in the \realcc{} setting to recovering a $c$ fraction of the codeword block (pertaining to the message symbols) in the \realdc{} setting, the amortized locality goes down by the same $c$ fraction, i.e., we need to query $(1/c)$ times more blocks.
\begin{restatable}{corollary}{arLCCToArLDC}\label{corr:arLCCToArLDC}
    Given an explicit code $\calC_{\realcc{}}$ that is a $(\alpha,\kappa, \delta,\eps)$-\realcc{} specified in Theorem~\ref{thm:FBDCtoARLCC} with rate $1 - \Rate{}$, for any $c \leq 1 - \rate{}$, there exists an explicit construction of code $\calC$ that is a $\left(\frac{\alpha} c,\kappa,\delta,\eps\right)$-$\realdc{}$ with rate $\frac{1 - \Rate{} - c}{1 - c}.$
\end{restatable}

%% file: sections/constconstruction.tex
In this section, we explicitly construct an ideal (constant rate, constant amortized locality, and constant error-tolerance) \realdc{} by instantiating the \FBDC{} in Section \ref{sect:FBDC} with an explicit expander code, converting the \FBDC{} into a \msf{arLCC} using an appropriate block code, and using the results of Section~\ref{sect:LCCtoLDC} to smoothly convert the \msf{arLCC} into a \realdc{}. 
For convenience, we restate the relevant statements.
\explicitExpander*
\vtester*
\recursiveExpander*
\topLevelAmp*
\FBDCtoARLCC*
\arLCCToArLDC*
We present an explicit, ideal \realdc{} in the following theorem.
\constantARLDC
\begin{proof}
Suppose $n = 2^r$, fix $\delta \in \Omega(1)$, let $B = log\left(1/\eps\right) \times (\log n)^2 \times 2^{(\log\log\log n)^4}$ be the lowest level block size (without loss of generality, suppose $B$ is also a power of $2$), and let $L = \log n \log\log n$.
Set $s = r - \lceil \log B \rceil - 1$ and define the sequence $\{(n_i,m_i)\}_{1 \leq i < s}$ as 
\begin{align*}
    n_i &= 2^{\lceil \log 2B \rceil + i - 1} \\
    m_i &= 2^{\lceil \log 2B \rceil - \lceil \log L \rceil + i - 1}
\end{align*}
for each $i \in [s -1].$
By Theorem~\ref{thm:explicitExpander}, there exists an explicit family of  $(d_{nest},\delta_{nest},\zeta_{nest})$ (unique) expanders $\{G_i = (L_i,R_i,E_i)\}_{1 \leq i < s}$ with $|L_i| = n_i$, $|R_i| = m_i$, left degree $d_i \leq d_{nest} = 2^{c_0 \log^3(\lceil \log  L \rceil )},$ left expansion size $\delta_i \geq \delta_{nest} =  \frac \beta {d_{i}2^{\lceil \log L \rceil}},$ and $\zeta \in \Omega(1)$.
Further, the average right degree is $\bar c_{nest} = d_{nest}\times 2^{\lceil \log L \rceil }.$

Then, by instantiating Proposition~\ref{prop:recursiveExpander} with the expander family above, setting $\eta_{nest} = 4/\zeta\beta$, and $g = \eta_{nest} \times \log(1/\eps)$ we construct a $(n,k)$ code $\calC'$ that is a $(\ell',B, \delta', \delta_{B},\eps')$-\FBDC{} where 
\begin{itemize}
    \item locality $\ell' = (s - 1) g\eta_{nest} \bar c _{nest} + 2^{\lceil \log 2B\rceil - 1}$,
    \item block size $2^{\lceil \log 2B\rceil - 1} = B$,
    \item top-level error tolerance $\delta' = \delta_{nest} / 2$,
    \item block-level error tolerance $\delta_B = 2\delta_{nest}$,
    \item soundness error $\eps' \leq  \exp\left(-g \left(\frac{\zeta \bar c_{nest} \delta_{nest}}{2} - \frac 1 {\eta_{nest}}\right)\right) =\exp\left( - g( \zeta\beta /2 -  \zeta\beta/4  )\right)= \eps,$
    \item and rate $1 - \Rate{}' \geq 1 - \sum_{i = 1}^{s - 1} m_i / n_i = 1 - (s - 1)\frac{1}{2^{\lceil \log L\rceil}} =  1 - \frac{(s - 1)}{L}.$
\end{itemize}

Since $\delta' \ll \delta$, apply Lemma~\ref{lem:topLevelAmp} using a \VLTC{} with constant error tolerance $\delta$. 
That is, let $n_{top} = n$ and $m_{top} = 2^{r - \lfloor \log (1/\delta) +\log \beta - c_0 \log^3(\log 1/ \delta)\rfloor}.$
Again by Theorem~\ref{thm:explicitExpander}, there is an explicit $(d_{top},\delta_{top},\zeta)$-unique expander $G_{top} = (L_{top},R_{top},\cdot)$ with $|L_{top}| = n_{top}$, $|R_{top}| = m_{top}$,
\begin{align*}
    \text{left degree } d_{top} &\leq 2^{c_0\log^3(\lceil \log 1/\delta \rceil )}, \\
    \text{left expansion size } \delta_{top} &\geq \frac{\beta}{d_{top} \times 2^{\lfloor \log (1/\delta) +\log \beta - c_0 \log^3(\log 1/ \delta)\rfloor}} \geq \delta, \tag{for large enough $n$}\\
    \text{average right degree } \bar c_{top} &= \frac{n_{top}d_{top}}{m_{top}} \leq 2^{\lfloor \log (1/\delta) +\log \beta \rfloor}.
\end{align*}
By Lemma~\ref{lem:vtester} and setting $\eta_{top} := \frac{8}{\zeta \delta_{nest} \bar c_{top}}$, we have a $\left(\ell_{top} = \eta_{top} \bar c_{top}, \delta, \rho_{top} = \zeta \bar c_{top} ,\eps_{top} =  \frac{1}{\eta_{top}} \right)$-\VLTC{}.
By Lemma~\ref{lem:topLevelAmp} with $g_{top} = \eta_{top}\log(1/\eps)$, we have an explicit construction of an $(n,k)$ code $\calC_{amp}$ that is a $(\ell, \delta,  2\delta_{nest}, \eps)$-\FBDC{}, where
    \begin{align*}
       \text{locality } \ell = g_{top}\ell_{top}+ \ell' &= \eta_{top}^2 \bar c_{top} + (s-1) g_{top}\eta_{nest}\bar c_{nest} + B  \\
        &= \log(1/\eps)\times \left(\frac{16}{\zeta^2 \delta_{nest}^2 \bar c_{top}}  +(s-1) \eta_{nest}\bar c_{nest}\right) + B \\
        &= \log(1/\eps) \times O(\delta / \delta_{nest}^2 + s d_{nest} L  + B) \tag{$1/{\bar c_{top}} = O(\delta)$}\\
        &= \log(1/\eps) \times O(d_{nest}^2 L^2) \tag{$
        s d_{nest}L\leq  1/ \delta_{nest}^2 = O(d_{nest}^2L^2)$}\\
        &= \log(1/\eps) \times 2^{O(\log^3\log\log n}) \times (\log n)^2 + B = (1 + o(1))\times B,
    \end{align*}
    the soundness error is at most $\max \left\{\eps', \exp\left(g\left(1/\eta_{top} -  \zeta\frac{\bar c_{top} \delta_{nest}}{4}  \right)\right)\right\} \leq \eps$ where 
        $g\left(\frac{1}{\eta_{top}} - \frac{\zeta \bar c_{top} \delta_{nest}}{4} \right) = g \times \left(-\frac{\zeta \bar c_{top} \delta_{nest}}{8}\right) = \log \eps,$
 and the rate $1 - \Rate{\FBDC{}}$ is at least 
\begin{align*}
    1 - \Rate{\FBDC{}} &\geq 1 - \frac{(s-  1)}{L} - m_s/n_s \\
    &= 1 - O(1/\log\log n) - \delta \times 2^{O(\log^3(\log 1 /\delta))} \geq \Omega(1)
\end{align*}
By Theorem~\ref{thm:FBDCtoARLCC}, take a $(B,k_{\msf{block}})$ code $\calC_{\msf{block}}$ with error tolerance at least $\delta_B \ll O(1)$ and constant rate (e.g., a Justesen code \cite{ITIT:Jus72}) to get a $\left(\alpha' = \frac{3\ell}{B} = O(1), B, \delta, \eps \right)$-\realcc{} with rate $1 - \Rate{\realcc{}} =1-  \Rate{\FBDC{}} - O(1)$ by Theorem~\ref{thm:FBDCtoARLCC}.

Finally, we can take $c = (1 - R_{\realcc{}})/2$ to get a $(\alpha =  O(2 \alpha' / (1 - \Rate{\realcc{}})) = O(1),\kappa =  B, \delta,\eps)$-\realdc{} with rate $1 - \rate{\realdc{}} \geq \frac{1  -\Rate{\realcc{}}}{1 + \Rate{\realcc{}}} \geq \Omega(1)$ by Corollary~\ref{corr:arLCCToArLDC}.
\end{proof}

%% file: sections/approach1construction.tex
In this section, we present two constructions of $\realdc{}$ with amortized locality less than $2$.
The first \realdc{} construction will have constant rate, constant error tolerance, and constant amortized locality less than $2,$ and the second \realdc{} construction will have sub-constant error tolerance, but rate and amortized locality approaching $1.$ 
Both of these constructions stand in significant contrast to the impossibility result of \rldc{}s with locality less than $2.$

Our constructions in this section can be seen as slight modifications of the parameters and codes chosen in the ideal \realdc{} construction in Section~\ref{sect:AllTogether}. 
Crucially, we will use a higher rate, lower error-tolerant (rate approaching $1$, $\Omega\left(\frac{\log\log n}{\log n}\right)$ error tolerance) for the lowest level blocks rather than the constant rate, constant error tolerance code used in the constant amortized locality construction, which we show can be constructed efficiently in Section~\ref{subsec:gvcodeapproach1}.
Additionally, we raise the $\kappa$ parameter, i.e., the minimum number of symbols recovered, to increase the number of lowest level blocks decoded. 
Along with a careful choice of parameter $c$ in the $c$-bad block filtering step to convert the \realcc{} into a \realdc{}, these modifications will result in a $\realdc{}$ with amortized locality less than $2$.

\subsection{Refining the \realdc{} Amortized Locality}
In this subsection, we present a modification of the prior \realdc{} construction in Theorem~\ref{thm:constantARLDC} to allow the amortized locality to reach less than $2$.
For reference, we restate Theorem~\ref{thm:FBDCtoARLCC} and Corollary~\ref{corr:arLCCToArLDC}.
We also state the following crucial Proposition~\ref{prop:boostedInner} about the existence of an error correcting code $\calC_{\msf{Block}}$ with error tolerance $\delta \geq \Omega\left(\frac{1}{\log n \log\log n}\right)$ and rate approaching $1.$
We will prove Proposition~\ref{prop:boostedInner} in the next Subsection~\ref{subsec:gvcodeapproach1} by showing the explicit construction of such a code.
\FBDCtoARLCC*
\arLCCToArLDC*

\begin{restatable*}{proposition}{boostedInner}\label{prop:boostedInner}
    For all $\eps > 0$, there exists a family of binary codes $\{\calC_k\}_{k \geq 1}$ such that any code in this family with codeword length $n \in \N$ has rate $k/n \geq \left(1 - \frac{1}{\log\log n} - \frac{(\log\log n)^3}{\log n}\right)$  and relative distance $\Omega\left(\eps \times \frac{\log\log n}{\log n}\right)$ 
    Furthermore code $\calC$ can be computed in time $\max\{n^{\eps},O(n)\}$. 
\end{restatable*}

\input{theorems/approachingOneARLDC}

\begin{proof}
Construct a \FBDC{} following the construction in Theorem~\ref{thm:constantARLDC} with block size $B = \log(1/ \eps) \times (\log n)^2 \times 2^{(\log\log\log n)^4}$ to get a  $(\ell_{\FBDC{}} = (\log(1/ \eps) \times(\log n)^2 \times 2^{(\log\log\log n)^4}, B, \delta, \delta_B = \Omega(1 / \log n\log\log n), \eps)$-\FBDC{} with rate $1 - \Rate{}' \geq 1 - O(1/\log\log n) - \delta \times 2^{c_0\log^3(\log 1 /\delta)}.$
There exists a constant $h = h(\nu) \geq \theta(1 / \nu)$ such that by Theorem~\ref{thm:FBDCtoARLCC} with code $\calC_{\msf{Block}}$ constructed in Proposition~\ref{prop:boostedInner}, we construct code $\calC_{\realcc{}}$ which is a $\left(\alpha_{\realcc{}}(\nu), \kappa(\nu), \delta, \eps \right)$ \realcc{} with $\alpha_{\realcc{}}(\nu) = 1 + \nu$, $\kappa(\nu) = \frac{\log(1/\eps) \times (\log n)^2 \times 2^{(\log\log\log n)^4}}{\nu},$ and rate $1 - \Rate{\realcc{}} \geq 1 - \Rate{}(\delta) - o(1).$
Finally, by Corollary~\ref{corr:arLCCToArLDC}, we take $c = 1 - \Rate{\realcc{}} - c_1$ to obtain a $(n,k)$ code $\calC$ that is a $(\alpha, \kappa, \delta,\eps)$ \realdc{} with amortized locality 
$\alpha \leq (1 + \nu) \times \frac{1}{1 - \delta \times 2^{c_0\log^3(\log 1 /\delta)} - c_1} + o(1)$ and rate $1 - \Rate{\realdc{}} \geq \frac{c_1}{c_1 + R(\delta)} - o(1)$.
\end{proof}
From Theorem~\ref{thm:approachingOneARLDC}, we may construct an ideal \realdc{} with amortized locality less than $2$. 
We will effectively set the rate of the \realcc{} greater than $1/2$, $1 - \rate{}(\delta) > 1/2,$ 
which will allow us to choose a $c$ parameter large enough so that the bad block filtering step does not increase the amortized locality over $2,$ i.e., $\alpha = \frac {\alpha_{\realcc}}{c} < 2$.
\lessthan
\begin{proof}
    Pick $c_1 = 1 - \frac{3(1 + \nu)}{5} - \rate{}(\delta).$
    Then, by Theorem~\ref{thm:approachingOneARLDC}, the amortized locality is \[\alpha = \frac{1 + \nu}{1 - \rate{}(\delta) - c_1} + o(1) = 5/3 + o(1),\]
    and the rate is 
    \[1 - \rate{\realdc{}} \geq \frac{1 - \frac{3(1+\nu)}{5} - \rate{}(\delta)}{1 - \frac{3(1 + \nu)}{5}} \in \Omega(1).\qedhere\]
\end{proof}
Additionally, by an alternate choice of parameters, we may achieve a \realdc{} with amortized locality approaching $1$. 
In particular, the amortized locality and rate are dependent on the error tolerance $\delta$ when $\nu > 0$ is fixed.
Let $\rate{}(\delta) = o(1)$ and pick $c_1 = c_2 - \Rate{}(\delta)$ for some $c_2 > \rate{}(\delta)$ while $c_2 = o(1)$.
Then, by Theorem~\ref{thm:approachingOneARLDC}, and a similar proof as in Corollary~\ref{corr:realdclessthan2}, we obtain a \realdc{} with rate and amortized locality approaching $1$.
\approach

For an example instantiation, take $\delta$ such that $1 - \Rate{}(\delta) \geq 1 - \frac{1}{ \log n}$ and set $c_1 = \frac{1}{\log n\log n} - \rate{}(\delta) .$
Then, for any fixed constant $\nu$, we can have amortized locality $\alpha = (1 + \nu) \times \frac{1}{1 - \Rate{}(\delta) + \Rate{}(\delta) -\frac{1}{\log n\log n}} \leq (1 + \nu) \times \frac{\log\log n}{\log\log n - 1} \leq (1 + \nu')$, for some constant $\nu' > 0$ (given sufficiently high $n$), and the rate is at least $\frac{ \frac{1}{\log \log n} - \frac 1 {\log n}}{\frac{1}{\log\log n}} - o(1) = 1 - o(1).$

We conclude this section with a comment about the possibility of a \realdc{} with simultaneous constant rate, constant error tolerance, and amortized locality approaching $1.$
In our construction we first construct a \realcc{} and then transform it into a \realdc{} via a bad block filtering step.
Unfortunately, this generic method of \realdc{} construction inherently limits the amortized locality to be at least the inverse of the rate, $\alpha \geq 1/(1-\rate{}).$
For a possible bad-case scenario, suppose that all blocks are $(1 - \rate{})$-good, i.e., the local rate of every block is equal to the global rate of the code.
Since the number of queries we make to recover a single codeword block is at least $B$ and the number of message symbols decoded from that block is at most $(1 - \Rate{}) B,$ the amortized locality can never be lower than the inverse rate $\frac 1 {1 - \Rate{}}$ in this scenario.

Hence, without any additional assumptions on the initial message smoothness of the \realcc{}, the amortized locality can only approach $1$ if the error tolerance approach $0$ in our construction.
It is an interesting question whether one can achieve a constant rate, constant error tolerance \realdc{} with amortized locality approaching $1$ by a direct construction.
We leave the construction of a constant rate, constant error tolerance \realdc{}s with amortized locality approaching $1$ as future work. 

\subsection{An Explicit Binary Code with Rate Approaching \texorpdfstring{$1$}{1}} \label{subsec:gvcodeapproach1}
We present an efficient construction of a binary code with rate $1 - o(1)$ and distance $\Omega\left(\frac{\log \log n}{\log n}\right)$. 
While such a code exists by the GV-bound \cite{ECT}, we employ code concatenation between a Reed Solomon outer code and a much smaller binary inner code with a rate-distance tradeoff close to the GV-bound to achieve our desired rate and distance parameters in any chosen polynomial time. 
We remark that the construction can even be made linear time with only a constant loss in distance.

Informally, recall that a concatenated code first encodes messages using an outer code over a large alphabet, then encodes each symbol with an inner code over a smaller alphabet (in our case, binary). 
The following Lemma states the achievable parameters of a concatenated code.
\begin{lemma}[Code Concatenation \cite{ECT}]\label{lem:concat}
    If there exists a $[N,K]$ code $\calC_{out}$ with rate $R_{out}$ and relative distance $\delta_{out}$ over an alphabet $\Sigma$ of size $2^{k}$ for some $n$ and there exists a $[n,k]$ binary code $\calC_{in}$ with rate $R_{in}$ and relative distance $\delta_{in}$, then there exists a $[Nn,Kk]$ binary code $\calC_{concat}$ with rate $R_{in} R_{out}$ and relative distance $\delta_{in}\delta_{out}.$

    Explicitly, let $\Enc_{out}$ and $\Enc_{in}$ be the encoding procedures of $\calC_{out}$ and $\calC_{in}$ respectively. 
    For any $\bx \in \Sigma^{K}$, let $\by = \Enc_{out}(\bx)$. 
    Then the encoding procedure $\Enc_{concat}$ of code $\calC_{concat}$ is defined as
    \[\Enc_{concat}(\bx) := \Enc_{in}(\by[1]) \circ \Enc_{in}(\by[2]) \circ \dots \circ \Enc_{in}(\by[N]),\]
    where each $\by[i] \in \Sigma$ is interpreted as a binary string of length $k.$ 
\end{lemma}
By Lemma~\ref{lem:concat}, we will need to choose an inner and outer code whose rates both approach $1$ with relative distances whose product is $\Omega\left(\frac{\log\log n}{\log n}\right)$.
We have flexibility in the choice of the outer code, which we choose to be an RS-code that achieves any rate $R$ with relative distance $(1 - R)$ (thereby reaching the singleton bound) with alphabet size $O(n).$
For our choice of binary inner code with codeword length $O(\log n)$, we recall the GV-bound for binary codes, where we derive a tight analysis on the runtime of finding such a code in Appendix~\ref{app:gvcode}.
Effectively, this implies that the time to construct such a code is asymptotically dominated by the other steps in our \realdc{} construction.

\begin{restatable}[Binary GV-bound \protect{\cite[Exercise~$4.6$]{ECT}}]{lemma}{binarygv}\label{lem:GV}
    For every $1/2 <  R \leq 1,$  there exists a family of binary codes $\{\calC_{k}\}_{k \in \N}$ with rate at least $R$ and relative distance $H^{-1}(1- R),$ where $H^{-1}$ is the inverse of the binary entropy function $H$ restricted to domain $[0,1/2).$
    Furthermore, any code $\calC$ in this family with codeword length $n$ can be found explicitly in $2^{d(n + k)}$ time for some constant $d \leq 3/2$.
\end{restatable}

\boostedInner
\begin{proof}
We apply code concatenation by defining outer code $\calC_{out}$ and inner code $\calC_{in}$. 
Let the outer code $\calC_{out}$ be a RS code with rate $\left(1 - \frac{1}{\log\log n}\right)$, relative distance $\frac{1}{\log\log n}$ with alphabet size $tn$ for some constant $t \leq 2.$

We construct the inner code $\calC_{in}$, as a binary code matching the GV-bound. 
Let $r = 1 - \frac{(\log\log n)^3}{2\log n}.$
By Lemma~$\ref{lem:GV}$, for any $c \geq 1$, we can compute a code $\calC$ with codeword length $\lceil{\frac{\log (tn)}{cr}}\rceil$, rate $r$, and relative distance 
\[\delta \geq H^{-1}(1 - r) = H^{-1}\left(\frac{(\log\log n)^3}{2\log n}\right),\]
in time at most $2^{d\lceil \log(tn) / c\rceil} \leq 2^{d}(tn)^{d/c} = O(n^{d/c})$ for some constant $d \leq 4$.
Pick $c \geq 2d/\eps = \theta(1/\eps)$, implying the time to construct $\calC$ is $O(n^\eps).$ 
Note that this additionally implies $c = \theta(1 / \eps).$

    Let $\Enc$ be the encoding procedure of code $\calC.$ 
    Then, define the encoding procedure $\Enc_{in}$ of inner code $\calC_{in}$ as
    \[\Enc_{in}(\bx) := \Enc(\bx_1) \circ \Enc(\bx_2) \circ \dots \circ \Enc(\bx_c),\]
    where $\bx = \bx_1 \circ \dots \circ \bx_c.$
    The rate of $\calC_{in}$ is the same as $\calC$, while the relative distance is $\delta / c$, since the codeword length of $\calC_{in}$ is a $c$ factor larger than the codeword length of $\calC.$

    By Lemma~\ref{lem:concat}, the concatenated code of $\calC_{out}$ and $\calC_{in},$ is an explicit binary code with rate \[\left(1 - \frac 1 {\log\log n}\right) \times \left(1 - \frac{(\log\log n)^3}{\log n}\right) \geq \left(1 - \frac{1}{\log\log n} - \frac{(\log\log n)^3}{\log n}\right)\]
    and relative distance $\geq \frac{1}{\log\log n} \times (\delta / c) \geq \frac{H^{-1}\left(\frac{(\log\log n)^3}{\log n}\right)}{c\log\log n} \geq \Omega\left(\eps \times \frac{ \log\log n}{\log n}\right)$.
    The last inequality follows from applying Theorem 2.2 in Calabro's thesis \cite{calabro2009exponential}, stating $H^{-1}(x) \geq \frac{x}{2\log (6/x)}$ for any $x \in [0,1/2)$, which implies
\begin{align*}
    H^{-1}\left(\frac{(\log\log n)^3}{2\log n}\right) &\geq \frac{(\log\log n)^3}{4\log n \times \left(\log (6\log n /(\log\log n)^3\right)} \\
    &= \frac{(\log\log n)^3}{4\log n\times (\log(6\log n) - 3\log\log\log n)} \\
    &\geq \frac{(\log\log n)^3}{4\log n\times (\log(6\log n) - 3\log\log\log n)}\\
    &\geq \Omega\left( \frac{(\log\log n)^2}{\log n}\right).
\end{align*}
\end{proof}

%% file: theorems/approachingOneARLDC.tex
\begin{restatable}[\realdc{} with Amortized Locality $< 2$ (or Approaching $1$)]{theorem}{approachingOneARLDC} \label{thm:approachingOneARLDC}
Let $\Rate{}(\delta) = \delta \times 2^{c_0 \log^3(\log 1/\delta)}$ for some universal constant $c_0 \in \N.$
    For any $n = 2^r$, $r \in \N$, $\delta \in (0,1], \eps < 1, c_1 \in (0,1 - \Rate{}(\delta)],$ and $\nu > 0$ there exists an (explicit) $(n,k)$ code $\calC$ that is a $\left(\alpha(\delta,\nu), \kappa(\nu) , \delta, \eps\right)$-$\realdc{}$, where $\alpha(\delta,\nu) = (1 + \nu) \times \frac{1}{1 - \Rate{} (\delta) - c_1} + o(1)$, $\kappa(\nu) = \frac{ \log(1/\eps)\times(\log n)^2\times 2^{(\log\log\log n)^4}}{\nu}$ and rate $1 - \rate{\realdc{}} \geq \frac{c_1}{c_1 + \Rate{}(\delta)} - o(1).$
\end{restatable}

%% file: sections/appendix.tex
\section{Finding GV-Bound Codes}\label{app:gvcode}
In this section, we present an efficient, deterministic algorithm for computing a binary, linear $(n,k)$ code meeting the GV-bound. 
This algorithm is taken from `Essential Coding Theory' textbook by Guruswami, Rudra, and Sudan \cite{ECT}.
We note that we include a finer analysis of the computation time.

We first define the well known notion of a Toeplitz matrix.
Intuitively, these are matrices whose left-to-right diagonals share the same value.
Denote the entry of matrix $A \in \{0,1\}^{k \times n}$ at the $i$'th row, $j$'th column as $A[i][j].$
\begin{definition}[Toeplitz matrix]
    A Toeplitz matrix $T \in \{0,1\}^{k \times n}$ satisfies the property that $T[i][j] = T[i +1][j+1]$ for all $i = 1,\dots, k - 1$ and $j = 1,\dots, n - 1.$
\end{definition}
We show that, with high probability, a randomly chosen Toeplitz matrix is the generator of a code reaching the GV-bound. 
\begin{lemma}\label{lem:randomToep}
    Let $T \in \{0,1\}^{k \times n}$ be a randomly chosen Toeplitz matrix. 
    For any $\bx \in \{0,1\}^k$ and any $\by \in \{0,1\}^n$, $\Pr[\bx T = \by] = 1/2^n.$
\end{lemma}
\begin{proof}
    Since $T$ is a randomly chosen Toeplitz matrix, it satisfies
    \[T[i][j] = r_{j - i} \tag{$i \in [k], j \in [n]$}\]
    for independent, random $\calR := \{r_{1 - k},r_{2 - k}\dots, r_{n - 1}\} \subset \{0,1\}^{n + k - 1}$.
    Let $\bx[m]$ be the last non-zero entry in $\bx.$
    Then for all $j \in [n],$
    \[(\bx T)[j] = \sum_{i = 1}^m \bx[i]r_{j - i} = \bx[m] r_{j - m} + \sum_{i = 1}^{m - 1} \bx[i]r_{j - i}.\]
    Next, observe that $r_{1 - m},\dots,r_{n - m}$ uniquely determines $\bx T$ (conditioned on all other variables $r_{i - j}$), i.e., $\Pr[\bx T = \by \mid \calR 
    \setminus \{r_{1 - m},\dots,r_{n - m}\}$. 
    By enumerating over all groups of conditioned variables, we have $\Pr[\bx T = \by] = (1/2)^n.$
\end{proof}
\binarygv*
\begin{proof}
    By Lemma~\ref{lem:randomToep}, for a randomly chosen Toeplitz matrix $T \in \{0,1\}^{k \times n}$, any nonzero $\bx \in \{0,1\}^k,$ and distance $\delta n$ for $\delta \in (0,1),$
    \[\Pr[\text{wt}(\bx T) < \delta n] = \frac{V(n, \delta n- 1)}{2^n},\]
    where $V(n,\delta n-1) = \sum_{i = 0}^{\delta n - 1} \binom{n}{i}$ is the volume of the (binary) hamming ball of radius $\delta n - 1$ and dimension $n.$
    Then, by a union bound over all $2^k - 1$ non-zero messages, 
    \[\Pr[\text{there exists } \bx \neq \textbf 0 \in \{0,1\}^k \text{s.t. wt}(\bx T) < d]] \leq (2^k - 1) \times \frac{V(n, \delta n- 1)}{2^n}.\]
    Since for any $V(n,\delta n) = 2^{H(\delta) n + o(n)},$ we do not find a GV-bound code generator (such that $k/n < 1 - H(\delta)$) with probability at most $2^{k - n + H(\delta) n + o(n)} \leq 2^{-\Omega(n)}$.
    Hence, by running a brute force search over all Toeplitz matrices, we will find at least one which reaches the GV-bound.

    Lastly, we show that running time is at most $2^{d(n + k)}$. 
    There are at most $2^{k + n - 1}$ Toeplitz matrices and we check if $\text{wt}(\bx T) \geq \delta n$ for at most $2^k  - 1$ nonzero $\bx \in \{0,1\}^k.$
    Hence, the total time for this algorithm is at most $2^{k + n - 1} \times (2^k  - 1) \times kn < 2^{2k + n} \leq 2^{(3/2)\times(n + k)}$, where the last inequality holds since $k \leq n.$
\end{proof}


\section{Impossibility of One-Query Relaxed Decoding with Perfect Completeness}\label{sect:onequeryRLDC}
In this section we show that one-query \rldc{}s do not exist. 
We mainly follow the original impossibility argument by Katz and Trevisan \cite{STOC:KatTre00} for one-query \LDC{}s.
Their technique was to extract a single-symbol encoding for a large fraction of `good' message symbols from the \LDC{}. Then, by using an information theoretic bound to show that the number of message symbols is necessarily constant, it follows that a single-symbol querying \LDC{} cannot exist.
To port this argument over to \rldc{}s, we additionally argue that any correct \rldc{} decoder must never output $\bot$.
Intuitively, on one query, the \rldc{} decoder will not be able to tell whether the queried symbol is corrupted or belongs to a codeword. 
Once this observation is made, the argument by Katz and Trevisan readily follows.

To start, we define \rldc{}s in a way convenient for the impossibility argument.
\begin{definition}[\rldc{}]
    A $(n,k)$ code $\calC = (\Enc,\Dec)$ (over $\Sigma$) is a $(\ell, \delta,\eps)$ relaxed locally decodable code (\rldc{}) if for every $\bx \in \Sigma^k$ we have
    \begin{enumerate}
        \item (Completeness)  For all $\bx \in \Sigma^k$ and $i \in [k]$, \[\Pr[\Dec^{\Enc(\bx)}(i) = x_i] = 1.\]
        \item (Soundness) For all $\bx \in \Sigma^k$, $i \in [k],$ and all words $\bw \in \Sigma^n$ such that $\reldist(\bw,\Enc(\bx)) \leq \delta,$ 
        \[\Pr[\Dec^{\bw}(i) \in \{\bot,x_i\}] \geq 1/2 + \eps.\]
        \item (Locality) For all $\bw \in \Sigma^n$ and all $i \in [k]$, $\Dec^{\bw}(i)$ makes at most $\ell$ queries to $\bw.$
    \end{enumerate}
\end{definition}
In the following lemma, we make the crucial observation that one-query \rldc{} decoders may not output $\bot$ due to the (perfect) completeness property.
\begin{lemma}[One-Query \rldc{} Decoders are \LDC{} Decoders]\label{lem:onequeryrldc}
    If $(n,k)$ code $\calC= (\Enc,\Dec)$ (over $\Sigma$) is a $(1, \delta,\eps)$ \rldc{}, then for all $\bw \in \Sigma^n$ and indices $i \in [k]$, 
    \[\Pr[\Dec^{\bw}(i) = \bot] = 0.\]
    This implies that for all words $\bw \in \Sigma^n$ such that $\reldist(\bw,\Enc(\bx)) \leq \delta,$ 
        \[\Pr[\Dec^{\bw}(i) = x_i] \geq 1/2 + \eps.\]
\end{lemma}
\begin{proof}
Without loss of generality, suppose the alphabet is binary, $\Sigma = \{0,1\}.$
    Since $\Dec$ makes at most one query, there exists a (deterministic function) $Q$ taking in as input index $i \in [k]$ and random coins $r$, and outputting a codeword index $j \in [n]$ such that $\Dec^{\bw}(i)$, on random coins $r$, queries $\bw[j]$.
    Similarly, there is a (possibly non-deterministic) decision rule $\Delta(i,\bw[j])$ which outputs $0,1,$ or $\bot.$
 
    Fix $i \in [k]$ and random coins $r$ of the decoder \Dec{}, hence fixing $j = Q(i,r).$ 
    Fix two words $\bx^0,\bx^1 \in \Sigma^k$ such that $\bx^0[i] = 0$ and $\bx^1[i] = 1.$ 
    Then, by perfect completeness, $\Dec^{\Enc(\bx^0)}(i) = 0$ and $\Dec^{\Enc(\bx^1)}(i) = 1$, implying
    \[\Enc(\bx^0)[j] \neq \Enc(\bx^1)[j] \text{ and so } \{\Enc(\bx^0)[j], \Enc(\bx^1)[j]\} = \Sigma.\]
    This implies that for all $i \in [k],$ $\bot \not\in \{\Delta(i,\sigma) : \sigma \in \Sigma\}$.
    Thus, $\Dec$ never outputs $\bot.$
\end{proof} 
We present the modified Katz and Trevisan impossibility argument for \rldc{}s.
\begin{lemma}[\cite{STOC:KatTre00}]\label{lem:ktCompress}
    Let $\Enc_{\star}: \Sigma^k \rightarrow \Sigma$ be a function. 
    Suppose there is an algorithm $\Dec_\star : \Sigma \times [k] \rightarrow \Sigma$ such that for all $i \in [k],$
    \[\Pr_{\bx \leftarrow \{0,1\}^k}[\Dec_\star(\Enc_\star(\bx),i) = x_i] \geq 1/2 + \eps.\]
    Then, $\log |\Sigma| \geq k \times (1 - \calH(1/2 + \eps))$, where $\calH$ is the $|\Sigma|$-ary entropy function.
\end{lemma}
\begin{theorem}
    Let $\calC$ be a $(n,k)$ code (over $\Sigma$) that is a $(1,\delta,\eps)$-\rldc{}. Then,
    \[k \leq \frac{1}{\delta \times (1 - \calH(1/2 + \eps))},\]
    where $\calH$ is the $|\Sigma|$-ary entropy function.
\end{theorem}
\begin{proof}
    We say codeword index $j \in [n]$ is good for message index $i \in [k]$ if 
    \[\Pr_{\bx \leftarrow \{0,1\}^k} [\Dec^{\Enc(\bx)}(i) = x_i \mid \Dec^{\Enc(\bx)}(i) \text{ queries $y_j$}] \geq 1/2 + \eps.\]
    We first show that for any message index $i \in [k],$ there are at least $\delta n$ good codeword indices for $i.$ 
    Observe, 
    \begin{align*}
        1/2 + \eps &\leq \Pr[\Dec^{\Enc(\bx)}(i) \in \{x_i,\bot\}] \\
        &= \Pr_{\bx }[\Dec^{\Enc(\bx)}(i) = x_i] \tag{Lemma~\ref{lem:onequeryrldc}}\\
        &= \sum_{j \in [n]} \Pr_{\bx}[\Dec^{\Enc(\bx)}(i) = x_i\mid  \Dec^{\Enc(\bx)}(i) \text{ queries $y_j$}] \Pr[\Dec^{\Enc(\bx)}(i) \text{ queries $y_j$}] 
    \end{align*}
     By an averaging argument, there exists an index $j_1$ good for $i$, i.e., \[\Pr_{\bx}[\Dec^{\Enc(\bx)}(i) = x_i\mid  \Dec^{\Enc(\bx)}(i) \text{ queries $y_j$}] \geq 1/2 + \eps.\]
     
    Let $\bd_{j_1,j_2,\dots} \in \{0,1\}^n$ be the perturbation vector that, on indices $j_1,j_2,\dots$, are a random bit, and $0$ on other indices.
    Let $\bw^1 := \Enc(\bx) \oplus \bd_{j_1}$.
    By the soundness of the \rldc{} and Lemma~\ref{lem:onequeryrldc}, 
    \begin{align*}
        1/2 + \eps &\leq \Pr[\Dec^{\bw^1}(i) \in \{x_i,\bot\}] \\
        &= \Pr_{\bx }[\Dec^{\bw^1}(i) = x_i] \tag{Lemma~\ref{lem:onequeryrldc}}\\
        &= \sum_{j \in [n]} \Pr_{\bx}[\Dec^{\bw^1}(i) = x_i\mid  \Dec^{\bw^1}(i) \text{ queries $y_j$}] \Pr[\Dec^{\bw^1}(i) \text{ queries $y_j$}] 
    \end{align*}
    Again, by an averaging argument, there exists an index $j_2$ good for $i.$
    Additionally, $j_2 \neq j_1$ since $\bw^1[{j_1}]$ is a random bit independent of the message $\bx.$ 

    We can apply this same argument $\delta n$ times by the soundness of the code to get $\delta n$ distinct indices $j_1,\dots,j_{\delta n} \in [n]$ such that each $j_p,$ is good for index $i$.
    Then, by applying pigeon hole principle, there exists a codeword index $j^* \in [n]$ that is good for at least $\delta k$ message indices $i_1,\dots,i_{\delta k} \in [k]$.
    In other words, there exists a function $\Enc_\star : \Sigma^{\delta k} \rightarrow \Sigma$ and an algorithm $\Dec_\star: \Sigma  \times [\delta k] \rightarrow \Sigma$ such that for all $i \in [\delta k], \Pr_{\bx}[\Dec_\star(\Enc_\star(\bx),i) = x_i] \geq 1/2 + \eps$.
    Thus, by Lemma~\ref{lem:ktCompress},
    $\log |\Sigma| \geq (1 - \calH(1/2 + \eps)).$
\end{proof}

%% file: sections/priorwork.tex
\paragraph{Amortized Locally Decodable Codes}
Blocki and Zhang introduce the notion of amortized locally decodable codes (\aldc{}) \cite{ISIT:blozha25}. 
Formally, a code $\calC = (\Enc,\Dec)$ is a  $(\alpha,\kappa,\delta,\eps)$-\aldc{} if (1) the local decoder $\Dec^{\bw}(L,R)$ given oracle access to a (possibly corrupted) codeword $\bw$ ($\reldist(\bw,\Enc(\bx)) \leq \delta$) and a sufficiently large interval range $R - L + 1 \geq \kappa,$ outputs $\bx[L,R]$ with probability at least $1 - \eps$ by making at most $\alpha (R - L + 1)$ queries.
They show that the Hadamard code gives rise to a $\aldc{}$ with amortized locality approaching $1$, in contrast to \LDC{} impossibility results. 
They additionally show that under the assumption of shared randomness or a resource-bound (e.g., computationally bounded, space-time bounded) channel, \aldc{}s with constant rate, constant error tolerance, and constant amortized locality exist (i.e., ideal \aldc{}s).   
In a follow-up work, Blocki and Zhang show that ideal \aldc{}s under the same assumptions exist when the (Hamming) errors are insertions and deletions \cite{ITC:bloZha25}. 
Their construction is based off a modification of a prior procedure that compiled Hamming \LDC{}s to insertion-deletion \LDC{}s and new constructions of shared-randomness \aldc{}s that satisfied a notion of ``consecutive-querying.''
This is even more surprising, given that \LDC{}s and its shared-randomness variant have provably worse tradeoffs for insertion and deletion errors than the tradeoffs for Hamming errors \cite{FOCS:BCGLZZ21, STOC:Gupta24}.

In this work, we provide the first information-theoretic construction of an \aldc{} by assuming a relaxation on the decoder.
It remains an open question whether we can construct ideal \aldc{}s without any relaxation or non-information-theoretic assumptions. 
Conversely, can one show that ideal \aldc{}s do not exist?
Another potential avenue of exploration would be the construction of amortized \rldc{}s for insertion and deletion errors.

\paragraph{Related Notions of Amortized Local Decoding.} The idea of generalizing local decodability to recovering a larger number of symbols has been explored previously, albeit without our focus on ideal parameters (constant rate, constant error tolerance, and constant amortized locality). 
Most similar to our study of ideal \aldc{}s is the work of Cramer, Xing, and Yuan that studies the multipoint recovery of Reed-Muller codes \cite{ITIT:CXY19}.
For a Reed-Muller code over a (prime) alphabet size $q > 4,$ they construct a local decoder that recovers any arbitrary subset of $\kappa$ symbols with $O(q^2 \kappa)$ queries to the codeword (with high probability $1 - \exp(-\Omega(\kappa))$). 
By choosing a constant alphabet size, their decoder achieves constant amortized locality, albeit with sub-constant rate. 
However, the Reed-Muller code cannot achieve constant rate and constant error tolerance simultaneously over a constant sized alphabet.
At the same time, increasing the alphabet size increases the amortized locality.
Hence, it remains an open question whether ideal \aldc{}s exist, even over non-constant size alphabets.

Batch codes are another similar code introduced by Kushilevitz, Ostrovsky, and Sahai, which consider decoding multiple message symbols from a tuple of $m$ encodings, where $t$ symbols are read from each encoding \cite{STOC:IKOS04}.
Specifically, the message $\bx$ is encoded into $(\by_1,\dots,\by_m)$. 
The decoder is given as input a subset of message indices $i_1,\dots,i_\kappa$ and must make at most $t$ queries to each $\by_j$ to recover $x_{i_1},\dots,x_{i_\kappa}$ with high probability.
While one can construct batch codes from \LDC{}s (via the equivalence between \LDC{}s and smooth codes \cite{STOC:KatTre00}), it is unclear if one can construct \aldc{}s from batch codes due to the absence of errors and fixed recovery size $\kappa$. 

Ramakrishnan and Wootters bring batch codes closer to our setting by studying a restricted notion of batch codes with erasures tolerance \cite{ISIT:RamWoo18}.
In particular, they consider the setting where the number of message symbols to recover $\kappa$ is equal to the number of queries made $\ell$. 
In this case, they show that the repetition code (if there are $d$ erasures, repeat each message symbol $(d+ 1)$ times) has optimal codeword length, even for large number of symbol requests $\kappa \in \Omega(k)$. 
This is surprising since for $\kappa = k,$ i.e., the traditional erasure code `global' decoding setting, codes reaching the singleton bound (e.g., a Reed-Solomon code) have optimal codeword length.
Setting the number of message symbols $\kappa$ to recover equal to the number of queries made $\ell$ is somewhat restrictive, and it is an interesting open question whether one can prove similar bounds when $\kappa \neq \ell.$
